\documentclass[final,1p]{elsarticle}
\usepackage{etoolbox}
\makeatletter
\patchcmd{\ps@pprintTitle}{\footnotesize\itshape
       Preprint submitted to \ifx\@journal\@empty Elsevier
      \else\@journal\fi\hfill\today}{\relax}{}{}
\makeatother

\makeatletter
\def\ps@pprintTitle{%
 \let\@oddhead\@empty
 \let\@evenhead\@empty
 \def\@oddfoot{\centerline{\thepage}}%
 \let\@evenfoot\@oddfoot}
\makeatother

\usepackage{amssymb}
\usepackage{amsmath}
\usepackage{color}

\usepackage{amsthm}
 \newtheorem{thm}{Theorem}[section]
 
 \newtheorem{lem}{Lemma}[section]
 \newtheorem{cor}{Corollary}[section]
 \newdefinition{rmk}{Remark}[section]
  \newdefinition{defe}{Definition}[section]
 \newproof{pf}{Proof}
 \newproof{pot}{Proof of Theorem \ref{thm2}}

\usepackage{caption}

\journal{Elsevier}

\begin{document}

\begin{frontmatter}



\title{Sampling and Reconstruction of Sparse Signals on Circulant Graphs - An Introduction to Graph-FRI}


\author[MSK]{M. S. Kotzagiannidis\fnref{label2}\corref{cor1}}
\ead{madeleine.kotzagiannidis@ed.ac.uk }

\address[MSK]{Institute for Digital Communications, The University of Edinburgh, King's Buildings, Thomas Bayes Road, Edinburgh EH9 3FG, UK}
\address[PLD]{Department of Electrical and Electronic Engineering, Imperial College London, London SW7 2AZ, UK}
\cortext[cor1]{Corresponding author}
\fntext[label2]{The work in this paper was carried out while the first author was a PhD student at the second author's institution.}
\author[PLD]{P. L. Dragotti}

\begin{abstract}
With the objective of employing graphs toward a more generalized theory of signal processing, we present a novel sampling framework for (wavelet-)sparse signals defined on circulant graphs which extends basic properties of Finite Rate of Innovation (FRI) theory to the graph domain, and can be applied to arbitrary graphs via suitable approximation schemes. 
At its core, the introduced Graph-FRI-framework states that any $K$-sparse signal on the vertices of a circulant graph can be perfectly reconstructed from its dimensionality-reduced representation in the graph spectral domain, the Graph Fourier Transform (GFT), of minimum size $2K$. By leveraging the recently developed theory of e-splines and e-spline wavelets on graphs, one can decompose this graph spectral transformation into the multiresolution low-pass filtering operation with a graph e-spline filter, with subsequent transformation to the spectral graph domain; this allows to infer a distinct sampling pattern, and, ultimately, the structure of an associated coarsened graph, which preserves essential properties of the original, including circularity and, where applicable, the graph generating set.
\end{abstract}

\begin{keyword}
graph signal processing \sep sampling on graphs \sep sparse sampling \sep graph wavelet \sep finite rate of innovation


\end{keyword}

\end{frontmatter}


\section{Introduction}
Contributions to \textit{Graph Signal Processing (GSP)} theory have aspired to create extensions of traditional signal processing notions to the graph domain, motivated by the need to gain a deeper understanding of how the complex connectivity of graphs may be leveraged for more sophisticated processing, computational efficiency and superior performance, all the while heading toward a more generalized theory of SP \cite{shu}. The inherent challenge of interpreting and incorporating newly arising data dependencies, while maintaining equivalencies to classical cases, has given rise to a variety of different approaches, borrowing notions from i.a. algebraic and spectral graph theory \cite{chung}, algebraic signal processing theory \cite{algmoura}, and general matrix theory \cite{golub}. \\
In its essence, the framework of GSP is concerned with the analysis or processing of (higher-dimensional) data naturally residing (or modelled) on the vertices of weighted graph structures, examples of which include transportation or social networks, with respect to the underlying network topology in an effort to exploit its inherent geometry.

A breadth of intriguing GSP problems such as graph wavelet analysis (\cite{Coifman}, \cite{ekambaram2}, \cite{ortega2}, \cite{ortega3}, \cite{spectral}), graph signal interpolation and recovery (\cite{interpol}, \cite{local}, \cite{mourarecov}), up to graph-based image processing (\cite{gwtim}, \cite{ortega2}) and semi-supervised learning (\cite{semi}, \cite{semisuper}), have been derived in the wake of two elementary model assumptions for the central graph operator: the positive semi-definite \textit{graph Laplacian matrix}, and the more generalized \textit{graph adjacency matrix}.\\
Nevertheless, on the path toward a generalized theory of signal processing, the field of GSP is just at the beginning, and the development of a rigorous theoretical foundation is required to fully understand and elucidate the potential of graphs.\\

In search of concrete analogies between traditional and graph SP, the class of \textit{circulant graphs} has been noted for its linear shift invariance property and provided the foundation for intuitive graph signal sampling and filtering operations, as first established in (\cite{ekambaram1}, \cite{ekambaram2}, \cite{Ekambaram3}), not least of all due to its characterization by the (permuted) DFT matrix, as an eigenbasis, in the spectral graph domain. 
This previously inspired our derivation of families of signal-sparsifying, vertex-localized and critically-sampled \textit{graph spline} and \textit{graph e-spline} wavelet filterbanks on circulant graphs (\cite{spie}, \cite{icassp}, \cite{splinesw}), with the vanishing moment property of the graph Laplacian operator and its parameterised generalization, the proposed \textit{e-graph Laplacian}, at its core. In particular, fundamental mathematical properties of the circulant graph Laplacian are detected and incorporated into novel generalized graph differencing operators, which further give rise to basis functions that are structurally similar to the classical discrete (e-)splines, as defined in \cite{espline}.
Equipped with reproduction and annihilation properties for higher-order complex exponential polynomial graph signals, these filterbanks can be iteratively applied for a sparse multiresolution signal representation on suitable coarsened graphs. For a thorough discussion of the underlying graph-based spline wavelet theory, we refer to our paper \cite{splinesw}.\\
\\
Given a sparse signal residing on the vertices of a circulant (or arbitrary) graph, it is desirable to exploit the sparsity property for sampling or dimensionality reduction, as conducted in the classical frame of signal processing or compressed sensing. In an effort to further pursue a widespread motivation to elucidate sparsity on graphs, this work addresses the problem of \textit{sparse sampling} and \textit{coarsening on graphs} by proposing an intuitive and comprehensive framework for sparse signals, characterized by a relatively small $l_0$-norm $||{\bf x}||_0=\#\{ i : x_i \neq 0\}$, residing on circulant graphs, as an extension of classical approaches in the Euclidean domain and which can further be generalized to arbitrary graphs by using the former as building blocks.\\ 
Complementary to our discussion of (e-)spline wavelets on circulant graphs in \cite{splinesw}, we proceed to investigate the sampling and recovery of sparse, and hence, \textit{wavelet-sparse} graph signals on circulant graphs, for which we derive a novel framework as a generalization of the traditional Finite Rate of Innovation (FRI) theory (\cite{vetorig}, \cite{vet}) to the graph domain. In particular, we show that, given its dimensionality-reduced spectral representation ${\bf y}$ in the graph Laplacian basis, the so-called Graph Fourier transform (GFT), a sparse graph signal ${\bf x}$ defined on the vertices of a circulant graph can be perfectly recovered using Prony's method \cite{vet}, while the coarse graph associated with the vertex-localized version of ${\bf y}$ is simultaneously identified, i.a. through a scheme of spectral sampling. We additionally refine and extend this approach to encompass (piecewise) \textit{smooth} graph signals, which have a sparse graph wavelet representation, including the sets of (piecewise) polynomials and complex exponential polynomials, in light of our newly derived constructions. Eventually, generalizations to (multi-dimensional) sampling on arbitrary graphs can be made i.a. on the basis of graph products, as we previously demonstrated for graph spline wavelet analysis \cite{splinesw}.\\
\\
{\bf Related Work:}\\
Signal recovery on graphs, denoting more broadly the empirical study as opposed to the analytical framework, has been tackled i.a. under the premise that a signal is smooth with respect to the underlying graph, and can for instance be formulated as an optimization problem in different settings (\cite{mourarecov}, \cite{dict}).
In \cite{ribeiro}, \cite{wiener}, \cite{ortegasampling}, and \cite{samplingmoura}, sampling theory for graphs, providing the specialized and more rigorous theorization of the former, is explored with predominant regard to the subspace of bandlimited graph signals under different assumptions; here, Anis et al. \cite{ortegasampling} and Chen et al. \cite{samplingmoura} provide two alternative interpretations of bandlimitedness in the graph domain, where, in particular, the latter uses matrix algebra to establish a linear reconstruction approach, based on the knowledge and suitable choice of the retained sample locations. Moving beyond the traditional domain, sampling theory in the context of graphs has furthermore attempted to address graph coarsening, as can be seen in \cite{samplingmoura}, also a problem in itself (\cite{coars1}, \cite{kron}), which bears the challenge of identifying a meaningful underlying graph for the sampled signal and has been generally featured to a lesser extent. \\
In particular, the graph coarsening scheme introduced in \cite{samplingmoura} by Chen et al. is comparable to the spectral-domain-based coarsening approach in our proposed Graph-FRI (GFRI) framework, up to the choice of the sampling set and resulting property preservations, and with the further distinction that we iteratively filter the given graph signal with a suitable graph e-spline filter prior to sampling. While the former requires $K$ entries of suitably chosen sample locations (for bandlimited signals of bandwidth $K$) for perfect recovery, our downsampling pattern is fixed and primarily used to identify the coarsened graph corresponding to the sampled graph signal, under preservation of certain graph properties, as well as independent of the reconstruction scheme, which solely requires the input of the dimensionality-reduced spectral graph signal ${\bf y}$. In addition, we consider sparse and graph wavelet-sparse, as opposed to bandlimited, graph signals, encompassing a wider variety of graph signal classes, which do not necessarily belong to a fixed subspace, as facilitated through suitable graph (e-)spline wavelet analysis. We first investigated the problem of sparse signal reconstruction on circulant graphs in \cite{globalsip} in the context of noisy recovery under (graph-)perturbations, with a preliminary discussion in \cite{icassp}.\\

Due to our focus on sparse graph signals, the comparison with compressive sensing (CS) \cite{comp} is imperative. In CS theory, a sparse signal ${\bf x}\in\mathbb{R}^N$ can be recovered with high probability from the dimensionality-reduced, sampled signal ${\bf y}={\bf A}{\bf x}$ under suitable conditions on the rectangular sampling operator ${\bf A}\in\mathbb{R}^{M\times N}$ with $M<<N$ and sparsity $K=||{\bf x}||_0$, by solving an $l_1$-minimization problem, or alternatively, using greedy reconstruction algorithms \cite{bp}.
While in contrast to compressive sensing approaches \cite{compr}, the recovery of the sparse vector ${\bf x}$ in our scheme is exact at the critical dimension of $2K$ measurements and based on a direct, spectral estimation technique, known as Prony's method (\cite{prony}, \cite{vet}), we note that neither requires knowledge of the locations of the non-zero entries. 
In addition, compressed sensing theory can be extended to the recovery of non-sparse signals ${\bf x}={\bf D}{\bf c}$ that have a sparse representation ${\bf c}$ in properly designed, overcomplete dictionaries ${\bf D}$ \cite{csdict}, which has also been addressed in the context of graphs by training a graph-based dictionary \cite{dict}. Our sampling framework takes a similar approach in that piecewise smooth (wavelet-sparse) graph signals ${\bf x}$ are filtered with a circulant multilevel graph wavelet transform in order to produce sparse signals ${\bf c}$ which can subsequently be sampled; nevertheless, the recovery of ${\bf x}$ from ${\bf c}$ ultimately follows from the invertibility of the wavelet transform.
\\
\\
In this work, we make the following main \textit{contributions}:
\begin{itemize}
\item A novel framework for the sampling and perfect reconstruction of sparse and graph-wavelet-sparse signals on circulant graphs (Thm. $4.1$)
\item A general scheme to extract the coarse graph associated with the sampled signal to accompany the above, including a property-preserving approach based on spectral sampling (Thm $4.2$)
\item Generalizations to sampling and recovery on arbitrary graphs, i.a. via graph product decomposition approximations
\end{itemize}
We summarize preliminaries in Section $2$, and provide an overview of our previously derived graph e-spline wavelet filterbank constructions with some novel results in Section $3$, before introducing the proposed sampling framework in Section $4$. Section $5$ features extensions to arbitrary graphs via graph product decompositions, and Section $6$ contains concluding remarks with motivations for future directions.

\section{Preliminaries}
A graph $G=(V,E)$, with vertex set $V=\{0,...,N-1\}$ of cardinality $|V|=N$ and edge set $E$, is characterized by an adjacency matrix ${\bf A}$, with $A_{i,j}>0$ if vertices $i$ and $j$ are adjacent, and $A_{i,j}=0$ otherwise, and its degree matrix ${\bf D}$, which is diagonal with entries $D_{i,i}=\sum_{j} A_{i,j}$. The combinatorial graph Laplacian ${\bf L}={\bf D}-{\bf A}$ of undirected graph $G$ is a positive semi-definite matrix, with a complete set of orthonormal eigenvectors $\{{\bf u}_{l}\}_{l=0}^{N-1}$ and associated non-negative eigenvalues $0= \lambda_0<\lambda_1\leq \dots \leq \lambda_{N-1}$, termed `graph-frequencies'. \\
We consider graphs that are undirected, connected, (un-)weighted, and do not allow self-loops; our primary focus however lies on the class of
circulant graphs due to their LSI (Linear Shift Invariance) property and regularity which facilitate a more intuitive application and extension of traditional signal processing concepts to the graph domain (examples of which can be seen in Fig. \ref{fig:sets}). 
A circulant graph $G$ with generating set $S=\{s_1,\dots,s_M\}$ and $0<s_k\leq N/2$, has adjacency relations between node pairs $(i,(i\pm s_k)_N),\forall s_k\in S$, for $mod N$ operator $()_N$, or alternatively, a graph is circulant under some node labelling if its associated graph Laplacian is a circulant matrix \cite{ekambaram1}. Further, the symmetric, circulant graph Laplacian matrix ${\bf L}$, with first row $\lbrack l_0\quad ... \quad l_{N-1}\rbrack$, has representer polynomial $l(z)=\sum_{i=0}^{N-1} l_i z^i$ with $z^{N-j}=z^{-j}$. The $2B$-regular ring lattice $G$, within a special sub-class of circulant graphs, has the generating set $S=\{1,...,B\}$, such that there is an edge between nodes $i$ and $j$, if $(i-j)_{N}\leq B$, and ${\bf L}$ is banded of bandwidth $B$.  Bipartite graphs, which are characterized by a vertex set $V=X\cup Y$ of two disjoint subsets $X$ and $Y$, such that no two vertices within the same set are adjacent, form another notable class of graphs within GSP.\\
In this work, we consider graph signals ${\bf x}$ residing on the vertices of a graph $G$ that are complex-valued, with sample value $x(i)$ at node $i$ and represented as the vector ${\bf x}\in\mathbb{C}^N$ \cite{shu}, while maintaining real weights between connections on $G$. The Graph Fourier Transform (GFT) $\hat{{\bf x}}$ of ${\bf x}$ defined on $G$, is the expansion in terms of the graph Laplacian eigenbasis ${\bf U}=\lbrack {\bf u}_0| \cdots |{\bf u}_{N-1}\rbrack$ such that $\hat{{\bf x}}={\bf U}^H {\bf x}$, where $H$ denotes the Hermitian transpose, extending the concept of the Fourier transform to the graph domain \cite{shu}. Notably, the GFT of circulant graphs can be expressed as a permutation of the DFT-matrix.\\
\begin{figure}[tbp]
	\centering
	
	{\includegraphics[width=1.2in]{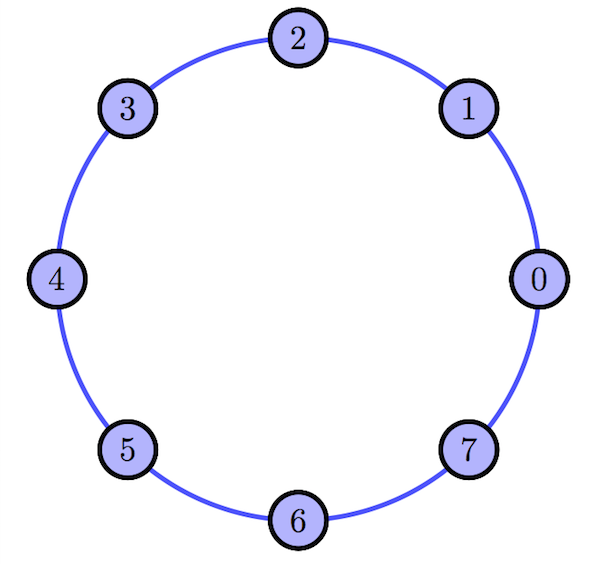}}
	{\includegraphics[width=1.23in]{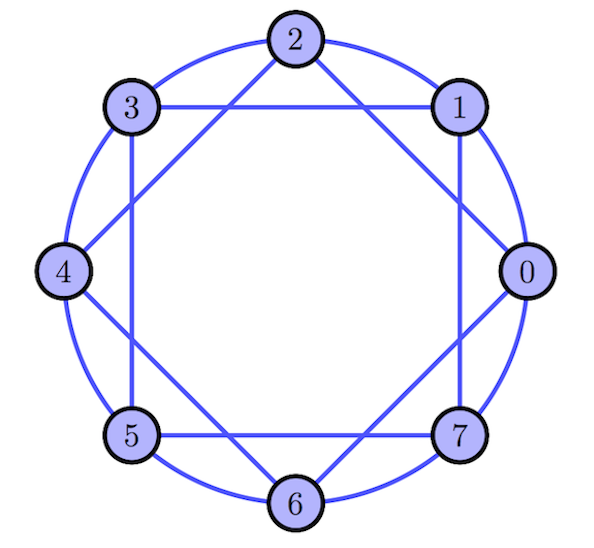}}
	{\includegraphics[width=1.23in]{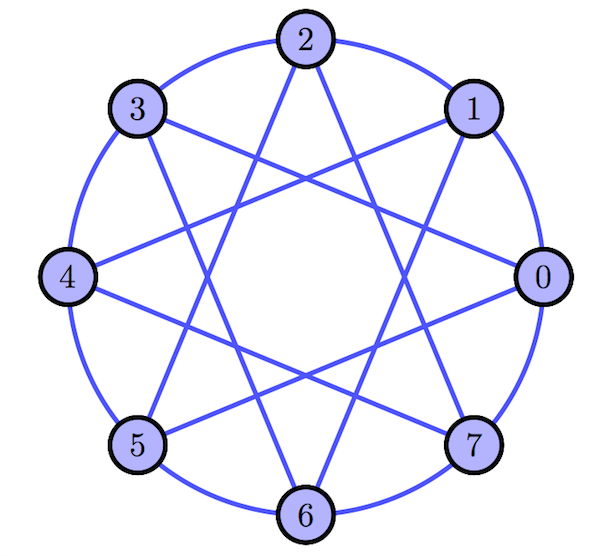}}
	{\includegraphics[width=1.2in]{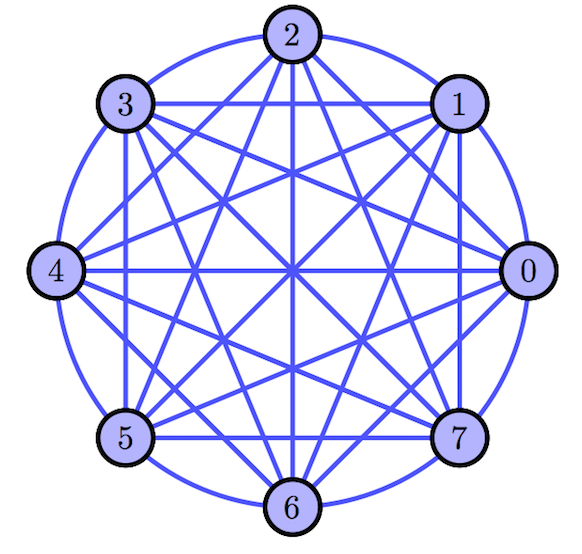}}
	\caption{Circulant Graphs with generating sets $S=\{1\}$, $S=\{1,2\}$, $S=\{1,3\}$ and $S=\{1,2,3,4\}$ (f. left)\label{fig:sets}}
\end{figure}

A graph (wavelet) filter ${\bf H}$ in the vertex domain generally describes a linear transform which takes weighted averages (differences) of components of the input signal ${\bf x}$ at a vertex $i$ within its $k$-hop local neighborhood $N(i,k)$, and may, where applicable, be expressed as a polynomial in the adjacency (or an alternative graph) matrix ${\bf H}=\sum_{k=0}^{N-1} h_k {\bf A}^k$ for appropriate coefficients $h_k$ \cite{shu}. Upon definition of a suitable set of low-and high-pass graph filters, along with a sensible (graph-dependent) downsampling pattern in the vertex domain, one may construct a graph wavelet filterbank for the graph at hand, with potential multiresolution analysis arising from the reassignment of the downsampled output to suitably coarsened graphs and iteration in the low-pass branch, as conducted e.g. in \cite{ekambaram1} for circulant graphs. Nevertheless, the overall task remains challenging in general due to the complex and variable connectivity of arbitrary graphs. \\
Furthermore, in \cite{ekambaram1}, a variety of SP concepts and operations on circulant graphs are discussed, including different options to conduct downsampling of vertices in the context of graph wavelet analysis. Here, a given signal on the circulant graph $G$ with generating set $S$ can be sampled by $2$ with respect to any element $s_k\in S$, and, for simplicity, we resort to the standard downsampling operation with respect to the outmost cycle ($s_1=1$), i.e. we skip every other labelled node, assuming that $G$ is connected such that $s_1\in S$, and $N=2^n$ for $n\in \mathbb{N}$. 
After downsampling, the retained vertices can be reconnected to form a coarsened graph for which several schemes have been proposed (\cite{kron},\cite{mult}). We primarily resort to opting for the sparsest possible graph-reconnection under the preservation of circularity, by either retaining the same generating set of the original circulant graph, or alternatively, reducing connectivity by only maintaining existing edges without newly reconnecting nodes (with the exception of preserving $s_1\in S$). In particular, this ensures that the bandwidth of the original graph adjacency matrix is not increased after coarsening, and, as a result of the relation between the sparsity of signal representation and the support (matrix bandwidth) of the proposed graph wavelet filters, thus facilitates a sparse multiresolution representation. As we will proceed to demonstrate in Sect. $4$, the former approach involving the replication of the graph generating set preserves basic graph properties and will be further leveraged in our sparse sampling scheme. 

\section{E-Spline Wavelet Analysis on Circulant Graphs}
Before we can formulate a framework for graph signal sampling, we need to state the theory of graph spline wavelets and their basic properties, which serve as a crucial element in our interpretation and analysis of sparsity on graphs. Inspired by the circulant graph wavelet filterbank introduced in (\cite{ekambaram2},\cite{Ekambaram3}) and in light of further detected properties pertaining to the circulant graph Laplacian matrix \cite{spie} and its parameterised extension, the e-graph Laplacian \cite{icassp}, which we will briefly state below, we have developed novel families of graph (e-)spline wavelet transforms which form a graph-analogy to the traditional (e-)spline and associated wavelet families. For a more thorough discussion of the comparison and proofs of the accompanying claims, we refer the interested reader to the comprehensive work on graph (e-)spline wavelets \cite{splinesw}.
\subsection{Vanishing Moments of the Graph Laplacian}
In the ensuing discussion, we distinguish between two main classes of smooth graph signals residing on the vertices of a graph $G$: \\
\\ \textit{(Piecewise) Polynomial:} A graph signal ${\bf p}\in \mathbb{R}^N$ defined on the vertices of a graph $G$ is (piecewise) polynomial if its labelled sequence of sample values, with value $p(i)$ at node $i$, is the discrete, vectorized version of a standard (piecewise) polynomial, such that ${\bf p}=\sum_{j=1}^K {\bf p}_j \circ {\bf 1}_{\lbrack t_j,t_{j+1})}$, where $t_1=0$ and $t_{K+1}=N$, with pieces $p_j(t)=\sum_{d=0}^D a_{d,j} t^d,\enskip j=1,...,K$, for $t\in\mathbb{Z}^{\geq 0}$, coefficients $a_{d,j} \in\mathbb{R}$, and maximum degree $D=deg(p_j(t))$.\\
\\ \textit{Complex exponential polynomial:} A complex exponential polynomial graph signal ${\bf x}\in\mathbb{C}^N$ with parameter $\alpha\in\mathbb{R}$, is defined such that node $j$ has sample value $x(j)=p(j) e^{i \alpha j}$, for  polynomial ${\bf p}\in\mathbb{R}^N$ of degree $\textit{deg}(p(t))$.\\
\\
Prior analysis of the graph Laplacian matrix has yielded a distinct annihilation property for the symmetric circulant case. Here, we adopt the traditional definition of the \textit{vanishing moments} of order $N$ of a high-pass filter ${\bf h}$ with taps $h_k$ as orthogonality of the former with respect to subspaces of polynomials of up to degree $N-1$, i.e. 
the $n$-th order moments $m_n=\sum_{k\in\mathbb{Z}} h_k k^n$ of ${\bf h}$, for $n=0,...,N-1$ are zero, in order to capture the following results on graph differencing operators:
\begin{lem} \label{lem1}For an undirected, circulant graph $G=(V,E)$ of dimension $N$, the associated representer polynomial $l(z)=l_0 +\sum_{i=1}^{B} l_{i} (z^i+z^{-i})$ of graph Laplacian matrix ${\bf L}$, with first row $\lbrack l_0 \ l_1\ l_2 \quad ... \quad l_2 \ l_{1}\rbrack$, has two vanishing moments. Therefore, the operator ${\bf L}$ annihilates polynomial graph signals of up to degree $D=1$, subject to a border effect determined by the bandwidth $B$ of ${\bf L}$, provided $2B<N$.\end{lem}
\noindent Further, we define a novel generalized graph difference operator, the e-graph Laplacian matrix, for undirected, circulant graph $G$ with adjacency matrix ${\bf A}$ of bandwidth $B$ and degree $d=\sum_{j=1}^B 2 d_j$ per node, with symmetric weights $d_j=A_{i,(j+i)_N}$, as $\tilde{{\bf L}}_{\alpha}=\tilde{{\bf D}}_{\alpha}-{\bf A}$, where $\tilde{d}_{\alpha}=\sum_{j=1}^B 2 d_j \cos(\alpha j)$ is the parameterised, exponential degree with $|\tilde{d}_{\alpha}|\leq d$ and $\alpha\in\mathbb{R}$. This operator can be considered as a generalization of the classical graph Laplacian, where $\tilde{{\bf L}}_{\alpha}={\bf L}$ for $\alpha=0$, and, although not a positive semi-definite matrix for $\alpha\neq 0$, it is of primary interest as a graph differencing operator, as the following property demonstrates:
\begin{lem} \label{lem2}For an undirected, circulant graph $G=(V,E)$ of dimension $N$, the associated representer polynomial $\tilde{l}_{\alpha}(z)=\tilde{l}_0 +\sum_{i=1}^{B} \tilde{l}_{i} (z^i+z^{-i})$ of the e-graph Laplacian matrix $\tilde{{\bf L}}_{\alpha}$, with first row $\lbrack \tilde{l}_0 \ \tilde{l}_1\ \tilde{l}_2 \quad ... \quad \tilde{l}_2 \ \tilde{l}_{1}\rbrack$, has two vanishing exponential moments, i.e. the operator $\tilde{{\bf L}}_{\alpha}$ annihilates complex exponential polynomial graph signals with exponent $\pm i\alpha$ and $\textit{deg}(p(t))=0$. Unless $\alpha=\frac{2\pi k}{N}$ for $k\in\lbrack 0, N-1\rbrack$, this is subject to a border effect determined by the bandwidth $B$ of $\tilde{{\bf L}}_{\alpha}$, provided $2B<N$.\end{lem}
\noindent Proofs of the preceding Lemmata entail the detection of roots $z_{\pm}=e^{\pm i\alpha}$ for representer polynomials $\tilde{l}_{\alpha}(z)$ (with double root $z=1$ for $\tilde{l}_0(z)=l(z)$), which indicate two exponential vanishing moments.

\subsection{Families of Graph E-Spline Wavelets}
By leveraging the aforementioned high-pass filter properties of the e-graph Laplacian operator, we design higher-order critically-sampled and vertex-domain localized graph wavelet filterbanks, which extend classical (e-)spline properties to the graph domain, and distinguish between graph spline and graph e-spline wavelets respectively:
\begin{thm} \label{esp1}Given the undirected, and connected circulant graph $G=(V,E)$ of dimension $N$, with adjacency matrix ${\bf A}$ and degree $d$ per node, we define the higher-order graph-spline wavelet transform (HGSWT), composed of the low-and high-pass filters
\begin{equation}\label{eq:lp1}{\bf H}_{LP}=\frac{1}{2^k}\left({\bf I}_N+\frac{{\bf A}}{d}\right)^k\end{equation}
\begin{equation}{\bf H}_{HP}=\frac{1}{2^k}\left({\bf I}_N-\frac{{\bf A}}{d}\right)^k\end{equation}
whose associated high-pass representer polynomial $H_{HP}(z)$ has $2k$ vanishing moments. This filterbank is invertible for any downsampling pattern, as long as at least one node retains the low-pass component, while the complementary set of nodes retains the high-pass components.\end{thm}
\begin{thm} \label{esp2}The higher-order graph e-spline wavelet transform (HGESWT) on a connected, undirected circulant graph $G$, is composed of the low-and high-pass filters
\begin{equation}\label{eq:lp2}{\bf H}_{LP_{\vec{\alpha}}}=\prod_{n=1}^T\frac{1}{2^k} \left(\beta_n{\bf I}_N+\frac{{\bf A}}{d}\right)^k\end{equation}
\begin{equation}{\bf H}_{HP_{\vec{\alpha}}}=\prod_{n=1}^T \frac{1}{2^k}\left(\beta_n{\bf I}_N-\frac{{\bf A}}{d}\right)^k\end{equation}
where ${\bf A}$ is the adjacency matrix, $d$ the degree per node and parameter $\beta_n$ is given by $\beta_n=\frac{\tilde{d}_{\alpha_n}}{d}$ with $\tilde{d}_{\alpha_n}=\sum_{j=1}^B 2 d_j \cos(\alpha_n j)$ and $\vec{\alpha}=(\alpha_1,...,\alpha_T)$. Then the high-pass filter annihilates complex exponential polynomials (of deg$(p(t))\leq k-1$) with exponent $\pm i \alpha_n$ for $n=1,...,T$. The transform is invertible for any downsampling pattern as long as the eigenvalues $\gamma_i$ of $\frac{{\bf A}}{d}$ satisfy $|\beta_n|\neq |\gamma_i|,\enskip i=0,...,N-1$, under either of  the sufficient conditions
 \\
$(i)$ $k\in2\mathbb{N}$, or \\
$(ii)$ $k\in\mathbb{N}$ and $\beta_n, T$ are such that $\forall \gamma_i, f(\gamma_i)=\prod_{n=1}^T(\beta_n^2-\gamma_i^2)^k> 0$ or $f(\gamma_i)< 0$.\\
If parameters $\beta_n$, are such that $\beta_n= \gamma_i$, for up to $T$ distinct values, the filterbank continues to be invertible under the above as long as $\beta_n\neq 0$ and at least $\sum_{i=1}^{T} m_i$ low-pass components are retained at nodes in set $V_{\alpha}$ such that $\{{\bf v}_{+i, k} (V_{\alpha})\}_{i=1, k=1}^{i=T, k=m_i}$ (and, if eigenvalue $-\gamma_i$ exists, complement $\{{\bf v}_{-i, k}({V_{\alpha}^{\complement}})\}_{i=1, k=1}^{i=T, k=m_i}$) form linearly independent sets, where $m_i$ is the multiplicity of $\gamma_i$ and $\{{\bf v}_{\pm i,k}\}_{k=1}^{m_i}$ are the eigenvectors respectively associated with $\pm \gamma_i$.\end{thm}
In general, we can iterate on the low-pass branch of either transform to obtain a multilevel representation defined on a collective of suitably coarsened graphs, however, as a consequence of the non-stationarity of the latter (see \cite{esplinewav} for the traditional case), modifications to parameters $\vec{\alpha}$ apply; here, we require the parameterization of $\tilde{d}_{\alpha'_n}$ by $\alpha'_n=2^j \alpha_n$ at level $j$ in order to preserve annihilation properties at the coarser scale.\footnote{Technically, one may describe both graph wavelet filterbank constructions in Thms. \ref{esp1} and \ref{esp2} as `non-stationary' in the sense that the representer polynomials of the respective graph filters at different levels are not necessarily dilates of one another, as a result of their dependence on the adjacency matrix. In particular, unless the coarsened graph, on which the downsampled low-pass output is defined, bears identical edge relations to the initial graph (e.g. when the generating sets are identical for $2B<N$), the representer functions will change with the graph. Nevertheless, the general structure of the filters as polynomials in the adjacency matrix only changes in Thm. \ref{esp2} due to the parameterization by $\{\beta_i\}_i$.}
\\
\\
The aforementioned transforms can be applied on any undirected circulant graph $G$, yet, we observe some noteworthy property distinctions between bipartite and non-bipartite circulant graph cases, as well as when $|\beta_n|= |\gamma_i|$ is satisfied for normalized adjacency matrix eigenvalue $\gamma_i$, and some $n$ and $i\in\lbrack 0\enskip N-1\rbrack$.
\subsubsection{Properties and Special Cases}
In prior work on (classical) generalized e-spline wavelets, it has been established that a scaling filter in the $z$-domain $H_j(z)$ at level $j$ can reproduce a function of the form $P(t) e^{\gamma_m t}$, with $deg P(t) \leq(L_m -1)$ for multiplicity $L_m$ of $\gamma_m$, if and only if the former is divisible by the term $R_{2^j \vec{\gamma}}(z)$, $\forall j\leq j_0-1$, where $R_{\vec{\gamma}}(z)=\prod_{m=1}^M(1+e^{\gamma_m} z^{-1})$, with $\vec{\gamma}=(\gamma_1,...,\gamma_M)^T\in\mathbb{C}^M$, and $H_j(z)$ has no roots of opposite sign, i.e. $H_j(z)$ satisfies the generalized Strang-Fix conditions for suitable $\vec{\gamma}$ (\cite{esplinewav}, Thm. $1$). \\
Complementing Lemmata \ref{lem1}--\ref{lem2} on the graph differencing operator, we can therefore further deduce that for a circulant and bipartite graph, which is characterized by all-odd elements $s_k$ in generating set $S$ for even cardinality $|V|=N$, the low-pass filters ${\bf H}_{LP}$ and ${\bf H}_{LP_{\vec{\alpha}}}$ in Eqs. (\ref{eq:lp1}) and (\ref{eq:lp2}) respectively reproduce (higher-order) polynomial and complex exponential polynomial graph signals. This is equivalently subject to a border effect that depends on the bandwidth $Bk$ of the filter matrices, provided $2Bk<N$; the complete proofs are presented in (\cite{splinesw}, Cors. $3.1$-$3.2$).  \\
\\
In addition, one can explicitly describe the frame bounds, and hence, the $l_2$-norm condition number of the \textit{HGESWT} for a bipartite circulant graph in terms of its spectrum and the given parameters:
\begin{cor} The condition number $C$ of the \textit{HGESWT} matrix ${\bf W}$ for a bipartite circulant graph, with downsampling conducted w.r.t. $s=1\in S$, can be expressed as $C=\sqrt{\frac{\lambda_{max}}{\lambda_{min}}}$, where $\lambda=\frac{1}{2}\left(\prod_{n=1}^T\frac{1}{2^{2k}}(\beta_n+\gamma)^{2k}+\prod_{n=1}^T\frac{1}{2^{2k}}(\beta_n-\gamma)^{2k}\right)$ for eigenvalues $\gamma$ of $\frac{{\bf A}}{d}$.
\end{cor}
\noindent \textit{Proof.} See Appendix $A1$.\\
\\
We further note a transition between notions of local and global signal annihilation via the derived generalized graph differencing operators on circulants. For $\alpha_k$ of the form $\frac{2\pi k}{N},\enskip k\in\lbrack 0\enskip N-1\rbrack$, the complex exponential graph signal ${\bf x}$ with values $x(j)=e^{\pm i \alpha_k j}$ is periodic and can be perfectly annihilated by $\tilde{{\bf L}}_{\alpha_k}$, and, in the bipartite case, also reproduced by suitable graph low-pass filters, parameterized by e-degree parameter $\tilde{d}_{\alpha_k}=\sum_{j=1}^B 2 d_j \cos\left(\frac{2\pi k j}{N}\right)$ without any border effects. In particular, the eigenvalues of normalized adjacency matrix $\frac{{\bf A}}{d}$ can be expressed as $\gamma_k=\sum_{j=1}^B 2 d_j \cos\left(\frac{2\pi k j}{N}\right)/d,\enskip k\in\lbrack 0\enskip N-1\rbrack$, thereby satisfying $\gamma_k=\frac{\tilde{d}_{\alpha_k}}{d}$ for the chosen $k$. Since ${\bf x}$ is an eigenvector of $\frac{{\bf A}}{d}$ corresponding to eigenvalue $\gamma_k$, the e-graph Laplacian matrix $\tilde{{\bf L}}_{\alpha_k}$ can be reinterpreted as a shifted version of ${\bf L}$ whose nullspace consists of (a subset of) its eigenvectors, in the alternative (normalized) representation $\left(\gamma_k {\bf I}_N-\frac{{\bf A}}{d}\right){\bf x}={\bf 0}_N$. \\
\\
The transform defined in Thm \ref{esp2} continues to be invertible if $|\beta_n|= |\gamma_i|, \enskip i\in\lbrack 0\enskip N-1\rbrack$, is satisfied, up to certain exceptions for $\beta_n$ (and hence $\alpha_n$), i.a. for the standard alternating downsampling pattern with respect to $s=1\in S$. The following corollaries capture such restrictions on the invertibility of the graph e-spline wavelet transform for certain $\alpha_n$, and are eventually illustrated through an example:
\begin{cor} \label{cor1} The HGESWT ceases to be invertible for any downsampling pattern as well as fails to reproduce certain graph signals when $\exists \alpha_i, \alpha_j$ in $\vec{\alpha}$ such that $\tilde{d}_{\alpha_i}=-\tilde{d}_{\alpha_j}$ for $\tilde{d}_{\alpha_l}$ of the form $\sum_{k=1}^B 2 d_k \cos\left(\frac{2\pi k l}{N}\right)$ for $l\in\lbrack 0\enskip N-1\rbrack$ and $2B< N$, including the case $\tilde{d}_{\alpha_i}=0$.\end{cor}
\begin{proof} In the case of a general circulant graph, we have, for $\alpha_i, \alpha_j$ as above, $H_{LP_{\alpha_i}}(z)=-H_{HP_{\alpha_j}}(z)$, or ${\bf H}_{LP_{\alpha_i}}=-{\bf H}_{HP_{\alpha_j}}$ in matrix form and vice versa, which leads to annihilation in the low-pass and reproduction in the high-pass branch. \\
Further, one cannot demonstrate linear independence of the eigenvectors associated with $\gamma_i,\gamma_j$, and hence invertibility of the \textit{HGESWT}, for any downsampling pattern, which follows from Thm \ref{esp2}; for brevity we refer to the complete proof in (Appendix $A.2$ , \cite{splinesw}) subject to necessary changes. When $\tilde{d}_{\alpha_i}=0$, the filterbank reduces to the normalized adjacency matrix $\frac{{\bf A}}{d}$ up to a sign per row (and its powers), which is singular if ${\bf A}$ is singular, while its representer polynomial contains the zero root.
\end{proof}
\begin{cor} \label{cor2} Let $\gamma^{DFT}=\{\gamma_i\}_i$ denote the DFT-ordered spectrum of $\frac{{\bf A}}{d}$ for $\frac{{\bf A}}{d}={\bf V}{\bf \Gamma}{\bf V}^H$ with ${\bf V}$ as the $N\times N$ DFT-matrix, and consider the HGESWT, with parameters of the form $\beta_k=\frac{\tilde{d}_{\alpha_{k}}}{d}$ for $\alpha_k$ in $\vec{\alpha}$.
When downsampling is conducted with respect to $s=1\in S$, the HGESWT ceases to be invertible if $\exists \alpha_i, \alpha_j$ in $\vec{\alpha}$ for $\frac{\tilde{d}_{\alpha_i}}{d},\frac{\tilde{d}_{\alpha_j}}{d}\in{\bf \gamma}^{DFT}$, with respective multiplicities at frequency positions in sets $M_i=\{i_k\}_k$ and $M_j=\{j_k\}_k$ in $\gamma^{DFT}$, and such that $\frac{\tilde{d}_{\alpha_i}}{d}=\gamma_i$ is located at position $(s+ N/2)_N, \enskip s\in M_j\cup M_i$, of the DFT-ordered spectrum (and vice versa for $\tilde{d}_{\alpha_j}$).  When the graph is bipartite, this condition becomes equivalent to that of Cor. \ref{cor1} for the fixed downsampling pattern.\end{cor}
\begin{proof}
Given parameters of the form $\beta_{i}=\frac{\tilde{d}_{\alpha_{i}}}{d}$ which are contained in the spectrum ${\bf \gamma}^{DFT}$ of $\frac{{\bf A}}{d}$, we distinguish between the eigenvalues $\gamma_i$ (and if existent, $-\gamma_i$) and their multiplicities, such that $|\beta_i|=|\gamma_i|$, with corresponding eigenvectors ${\bf V}_{\pm\gamma_i}=\{{\bf v}_{\pm i,l}\}_l$.
Then the invertibility of the \textit{HGESWT} is conditional upon the eigenvector sets ${\bf V}_{\pm\gamma_i}$ respectively being linearly independent after downsampling each vector by $2$ to give ${\bf v}_{+i,l}(V_{\alpha})$, with $V_{\alpha}=\{0:2:N-2\}$ as the set retained nodes (for the detailed proof see \cite{splinesw}). Since ${\bf V}$ is the DFT-matrix and ${\bf V}(V_{\alpha},0:N-1)=\lbrack \tilde{{\bf V}}\enskip\tilde{{\bf V}}\rbrack$ with $\tilde{{\bf V}}$ as the DFT of dimension $N/2$ (up to a normalization constant), we observe that eigenvector pairs $({\bf v}_k,{\bf v}_{k+ N/2})$, at position $k\in\lbrack 0\enskip N-1\rbrack$ become linearly dependent. We therefore need to ensure that the parameters $\{\beta_i\}_i$ with $|\beta_i|=|\gamma_i|$ are chosen such that the corresponding values of $\{\gamma_i\}_i$ (and multiplicities) respectively do not take the aforementioned positions in the DFT-ordered spectrum; for existing $-\gamma_i$, the same relation holds for complement ${\bf V}_{-\gamma_i}(V_{\alpha}^{\complement})=\{{\bf v}_{-i,l}(V_{\alpha}^{\complement})\}_l$. When the graph is additionally bipartite, we note that given $\alpha_i,\alpha_j$ at respective positions $i,j$, with $j=(i+N/2)_N$, due to the relation $\cos\left(\frac{2\pi k (i+N/2)}{N}\right)=-\cos\left(\frac{2\pi k i}{N}\right)$ for odd $k$, we have $\tilde{d}_{\alpha_i}=-\tilde{d}_{\alpha_j}$ and Cor. \ref{cor1} applies. \end{proof}

\noindent {\bf Example:} Consider the unweighted bipartite circulant graph $G=(V,E)$ of dimension $N=|V|=64$ with generating set $S=\{1,3,5\}$ and normalized adjacency matrix $\frac{{\bf A}}{d}$. Define one level of the graph e-spline wavelet transform of Thm. \ref{esp2} on $G$ with parameters $\alpha_1=\frac{2\pi 15}{N}$ and $\alpha_2=\frac{2\pi 17}{N}$ and $k\in2\mathbb{N}$, tailored to the reproduction/annihilation of complex exponential signals with sample $y(t)=e^{\pm i\alpha_j t},\enskip j=1,2$ at node $t$. The transform is expressed as ${\bf W}=\frac{1}{2}({\bf I}_N+{\bf K}){\bf H}_{LP_{\vec{\alpha}}}+\frac{1}{2}({\bf I}_N-{\bf K}){\bf H}_{HP_{\vec{\alpha}}}$ for diagonal downsampling matrix ${\bf K}$, with $K(i,i)=1$ when node $i$ retains the low-pass component and $K(i,i)=-1$ otherwise.\\
We observe that the normalized e-degrees take the form $\beta_1=\frac{\tilde{d}_{\alpha_1}}{d}=0.093$ and $\beta_2=\frac{\tilde{d}_{\alpha_2}}{d}=-0.093$, which correspond to a pair of eigenvalues of $\frac{{\bf A}}{d}$, characterizing the spectral folding phenomenon of the bipartite graph spectrum \cite{chung}. Upon diagonalization by the DFT-matrix, we further observe that multiplicities of the former respectively occur at (frequency) positions $k_1=49$ and $k_2=47$ of the DFT-ordered spectrum of $\frac{{\bf A}}{d}$. For eigenvalues $0.093$ and $-0.093$ with respective frequency parameterizations $\frac{2\pi 15}{N}$ $\left(\text{or}\enskip \frac{2\pi 49}{N}\right)$ and $\frac{2\pi 47}{N}$ $\left(\text{or}\enskip\frac{2\pi 17}{N}\right)$, this implies $(15+N/2)_N=47$ and $(17+N/2)_N=49$, which according to Cors. \ref{cor1} and \ref{cor2}, violates the invertibility property of the \textit{HGESWT} when downsampling is conducted w.r.t. $s=1\in S$, and, more generally, for any downsampling pattern.  
In particular, we have ${\bf H}_{LP_{\vec{\alpha}}}={\bf H}_{LP_{\alpha_1}}{\bf H}_{LP_{\alpha_2}}={\bf H}_{HP_{\alpha_2}}{\bf H}_{HP_{\alpha_1}}={\bf W}$ and it can be easily deduced that high-pass filter ${\bf H}_{HP_{\vec{\alpha}}}$ is not invertible as its nullspace is non-empty. Fig. \ref{fig:gsignal} depicts the graph along with the associated normalized graph filter function.
\begin{figure}[tbp]
	\centering
	{\includegraphics[width=1.5in]{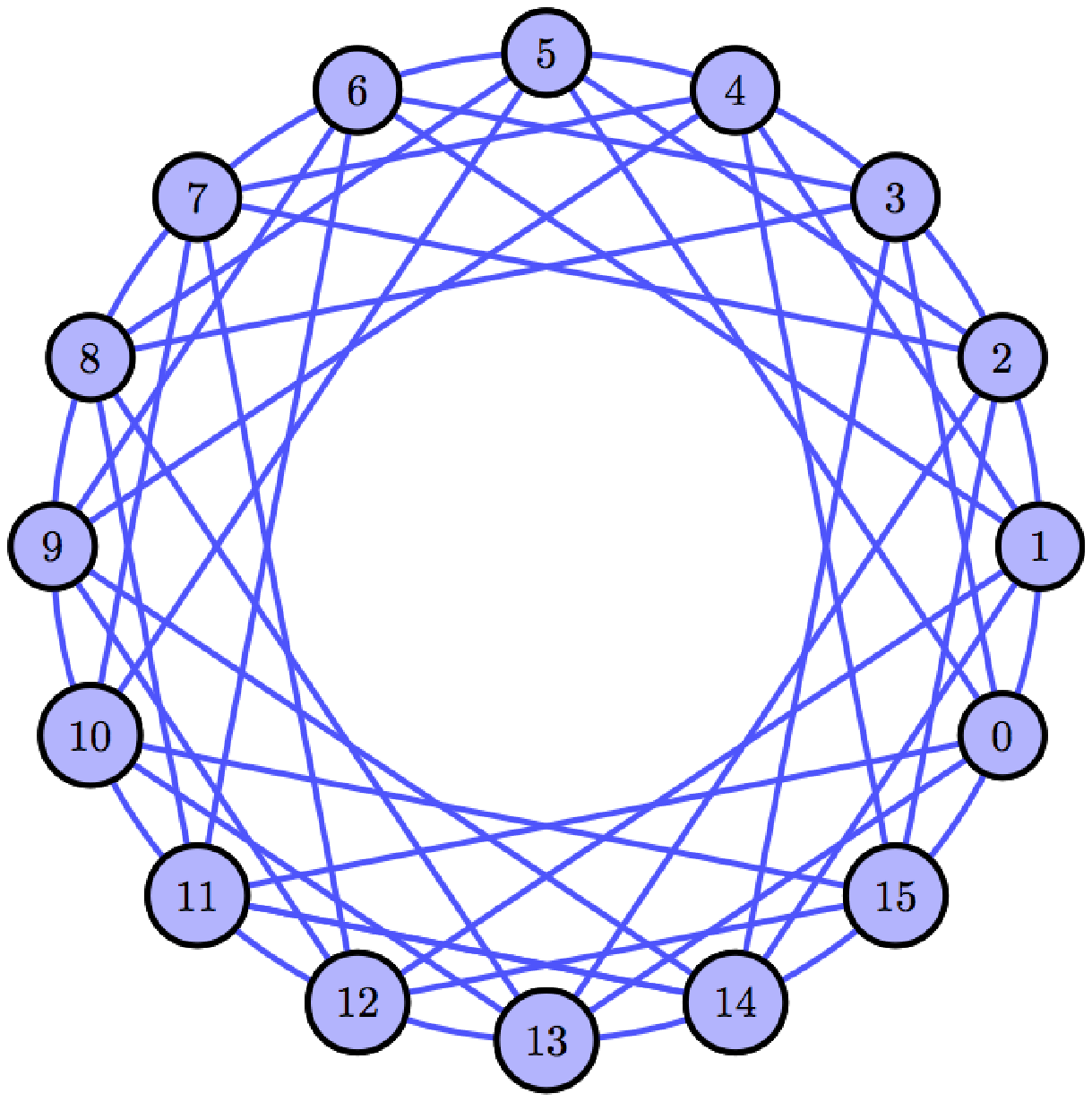}}%
	{\includegraphics[width=3.2in]{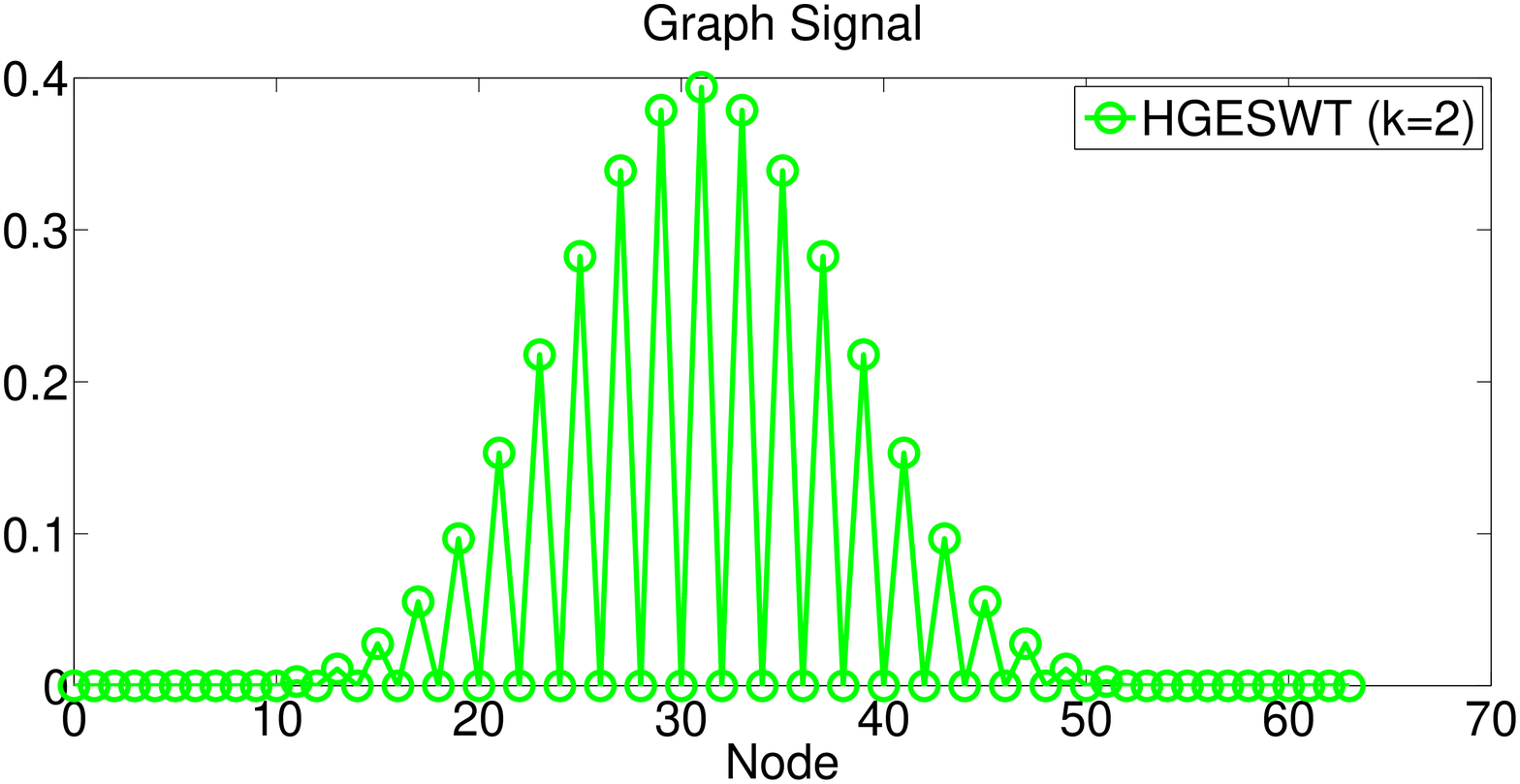}}%
		\caption{Illustrative Circulant Graph with $S=\{1,3,5\}$ of dimension $N/4$ (left) and associated normalized graph filter function of the \textit{HGESWT} ($k=2$, $\vec{\alpha}=(\alpha_1, \alpha_2)$) in the graph vertex domain of dimension $N=|V|=64$: centered at vertex $v=31\in V$ (right). \label{fig:gsignal}}
\end{figure}

\subsubsection{Complementary Graph E-Spline Wavelets for Non-Bipartite Circulants}
For non-bipartite circulant graphs, we resort to traditional spectral factorization techniques to create (vertex-localized) low-pass filters with the required exponential polynomial reproduction properties while maintaining the high-pass filter as is. These novel transforms are composed of well-defined analysis and synthesis filters of compact support and can be related to the previous filterbanks via a symmetric circulant transformation filter ${\bf C}$, depending on the invertibility of the low-pass filters in Eqs. (\ref{eq:lp1}) and (\ref{eq:lp2}) \cite{splinesw}.\\
In particular, given analysis high-pass filter $H_{HP_{\vec{\alpha}}}(z)=\prod_{n=1}^T\frac{\tilde{l}_{\alpha_n}(z)^k}{(2d)^k}$ with parameterization $\vec{\alpha}=(\alpha_1,...,\alpha_T)$, we define synthesis low-pass filter $\tilde{H}_{LP_{\vec{\alpha}}}(z)=H_{HP_{\vec{\alpha}}}(-z)$, and derive analysis low-pass filter $H_{LP_{\vec{\alpha}}}(z)$ from \begin{equation}\label{eq:fb}
P(z)=H_{LP_{\vec{\alpha}}}(z) \tilde{H}_{LP_{\vec{\alpha}}}(z), \quad\text{subject to}\enskip P(z)+P(-z)=2.\end{equation} Further constraints for reproduction properties are imposed in the form of (e-)spline factors in $H_{LP_{\vec{\alpha}}}(z)=\prod_{n=1}^T(z+2\cos(\alpha_n)+z^{-1})^k R(z)$, where $R(z)$ is determined to satisfy Eq. (\ref{eq:fb}). 
\begin{thm} \label{esp3} Given the undirected, and connected circulant graph $G=(V,E)$ of dimension $N$, with adjacency matrix ${\bf A}$ and degree $d$ per node, we define the higher-order `complementary' graph e-spline wavelet transform (HCGESWT) via the set of analysis filters:
\begin{equation}\label{eq:lp3} {\bf H}_{LP_{\vec{\alpha}},an}\stackrel{(*)}{=}{\bf C}\bar{{\bf H}}_{LP_{\vec{\alpha}}}={\bf C}\prod_{n=1}^T\frac{1}{2^k}\left(\beta_n {\bf I}_N+\frac{{\bf A}}{d}\right)^k\end{equation}
\begin{equation}{\bf H}_{HP_{\vec{\alpha}},an}=\prod_{n=1}^T\frac{1}{2^k}\left(\beta_n{\bf I}_N-\frac{{\bf A}}{d}\right)^k\end{equation}
and the set of synthesis filters:
\begin{equation}{\bf H}_{LP_{\vec{\alpha}},syn}=c_1{\bf H}_{HP_{\vec{\alpha}},an} \circ {\bf \mathit{I}}_{HP}\end{equation}
\begin{equation}{\bf H}_{HP_{\vec{\alpha}},syn}=c_2{\bf H}_{LP_{\vec{\alpha}},an} \circ {\bf \mathit{I}}_{LP}\end{equation}
where ${\bf H}_{LP_{\vec{\alpha}},an}$ is the solution to the system from Eq. (\ref{eq:fb}) for $\vec{\alpha}$ under specified constraints, with coefficient matrix ${\bf C}={\bf H}_{LP_{\vec{\alpha}},an}{\bf \bar{H}}_{LP_{\vec{\alpha}}}^{-1}$ in $(*)$ arising where applicable (see Cor. $3.3$ \cite{splinesw}). Here, $c_i,i\in\{1,2\}$ are normalization coefficients, ${\bf \mathit{I}}_{LP/HP}$ are circulant indicator matrices with first row of the form $\lbrack 1 \enskip-1\enskip 1\enskip -1\enskip...\rbrack$, and $\circ$ denotes the Hadamard product. \end{thm}
The existence of a suitable low-pass filter ${\bf H}_{LP_{\vec{\alpha}},an}$ to complete the above filterbank is conditional upon satisfying B\'{e}zout's Theorem, i.e. the representer polynomial $H_{HP_{\vec{\alpha}},an}(z)$ of ${\bf H}_{HP_{\vec{\alpha}},an}$ must yield no opposing or zero roots for given $\vec{\alpha}$ (\cite{splinesw}, \cite{esplinewav}). At $\vec{\alpha}={\bf 0}$, this gives rise to a regular graph spline wavelet filterbank. We note that while this design is also applicable to bipartite circulant graphs, it is less relevant, as the transform given in Thm. \ref{esp2} already provides the desired reproduction properties in that case. As an interesting aside, given a bipartite circulant graph, the special case from Cor. \ref{cor1} can be further extended to provide scenarios which violate B\'{e}zout's Thm. and thus the existence of a complementary filterbank construction when downsampling is conducted with respect to $s=1\in S$; specifically, when $\exists \alpha_i, \alpha_j$ in $\vec{\alpha}$ such that $\tilde{d}_{\alpha_i}=-\tilde{d}_{\alpha_j}$ for $\tilde{d}_{\alpha}$, the representer polynomial of the graph Laplacian product ${\bf H}_{HP_{\alpha_i}}{\bf H}_{HP_{\alpha_j}}$ contains opposing roots due to $H_{HP_{\alpha_j}}(z)=-H_{HP_{\alpha_i}}(-z)$.

\section{Sampling on Circulant Graphs}
The process of sampling a continuous time signal $x(t)$ in the Euclidean domain traditionally comprises filtering with a given $h(t)$, followed by a (uniform) sampling step, which creates the samples $y_n=(x*h)(t)|_{t=nT}$ at sampling rate $f_s=\frac{1}{T}$ for period $T$ (\cite{vet}, see Fig. \ref{fig:tss}). At its core, sampling theory provides a bridge between continuous-time  and discrete-time signals by seeking to identify ideal methods as well as conditions for the perfect recovery of $x(t)$ from the given $y_n$; this further extends to identifying distinct classes of $x(t)$ and suitable filters $h(t)$ which guarantee perfect recovery.\\
\begin{figure}[htb]\centering
 \centerline{\includegraphics[scale=0.4]{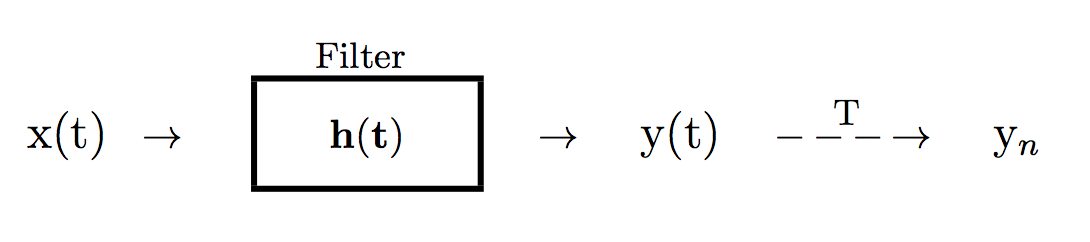}}
\caption{Traditional Sampling Scheme.}
\label{fig:tss}
\end{figure}
In a broader sense, sampling in discrete-time can be understood as a dimensionality reduction, which is followed by a dimensionality increase (or interpolation) to recover the original signal \cite{fourier2}. In order to formulate a sampling theory in the graph setting, additional questions need to be addressed, in particular, on what kind of graph structure the dimensionality-reduced (sampled) signal is defined and how it relates to the original graph, which challenges the classical problem and invites a more sophisticated take on sampling. Sampling theory on graphs can therefore be described as the study of methods and conditions which facilitate the perfect recovery of a graph signal ${\bf x}\in\mathbb{R}^N$ on the vertices of a graph $G$, with $|V|=N$, from a dimensionality-reduced, possibly graph-filtered signal ${\bf y}\in\mathbb{R}^M$, extending to the identification of a coarsened graph $\tilde{G}$, with $|\tilde{V}|=M$ and $M<N$, on which ${\bf y}$ is defined. Further intriguing extensions may involve the accompanying recovery of $G$ from $\tilde{G}$, however, in this work we restrict our focus on the former.\\
In the following analysis, we tie on the established graph spline wavelet theory to firstly describe wavelet-sparse signals, as an extension of the class of sparse signals, and, secondly, to formulate a comprehensive sparse signal sampling and graph coarsening framework.

\subsection{Sparsity and Sampling}
Equipped with novel families of sparsifying graph wavelet transforms, we proceed to explore sparse representations on circulant graphs by a priori defining smooth, or \textit{wavelet-sparse} graph signals. Let ${\bf x}_W\in\mathbb{C}^N$ be a signal defined on a circulant graph $G$ and ${\bf W}_j\in\mathbb{R}^{N/2^j\times N/2^j}$ represent a general Graph Wavelet Transform (GWT) of the form 
\[{\bf W}_j=\begin{bmatrix} {\bf \Psi}_{j\downarrow 2}{\bf H}_{LP_j}\\{\bf \Phi}_{j\downarrow 2}{\bf H}_{HP_j}\end{bmatrix},\]
composed of low-and high-pass filters ${\bf H}_{LP_j}, {\bf H}_{HP_j}\in\mathbb{R}^{N/2^{j}\times N/2^j}$, at level $j$, where the binary downsampling matrices ${\bf \Psi}_{j\downarrow 2}, {\bf \Phi}_{j\downarrow 2}\in\mathbb{R}^{N/2^{j+1}\times N/2^j}$ sample complementary sets of nodes in the standard alternating pattern with respect to $s=1\in S$; here, even-numbered nodes are retained in the low-pass branch and subsequently redefined on a suitably coarsened graph.. The multiresolution representation of ${\bf x}_W$, following iteration on the low-pass branch, then yields
\[\tilde{{\bf x}}={\bf W}{\bf x}_W=\begin{bmatrix}{\bf W}_j &\\ & {\bf I}_{\frac{N(2^j-1)}{2^j}}\end{bmatrix}\dots\begin{bmatrix}{\bf W}_1 &\\ & {\bf I}_{\frac{N}{2}}\end{bmatrix} {\bf W}_0 {\bf x}_W,\]
where ${\bf W}$ is the multilevel graph wavelet transform matrix. In order to redefine $\tilde{{\bf x}}$, whose individual partitions reside on a collective of coarsened graphs, with respect to the original $G$, we introduce the permutation matrix ${\bf P}$, so that for an appropriate relabelling ${\bf x}={\bf P}\tilde{{\bf x}}$ resides on $G$. 
\\
\\
Hence, we define the class $W$ of wavelet-$K$-sparse graph signals, with ${\bf x}_W\in W$ of dimension $N$, through the $K$-sparse multiresolution representation ${\bf x}\in\mathbb{C}^N$, $||{\bf x}||_0=K$, via a suitable GWT. In particular, given smooth graph signal ${\bf x}_W$, we can describe, and hence tailor, the sparsity $K$ of ${\bf x}$ as a function of decomposition level $j$ as well as of the bandwidth $B_j$ of the graph filter matrix at each $j$.\\
We present results assessing the number of non-zero entries of ${\bf x}\in\mathbb{C}^N$, for ${\bf x}_W$ in the class of polynomials, which are generalizable to complex exponential polynomials on circulant graphs. Here, we conduct node reconnection such that the coarse graphs retain their original generating set after downsampling at each level, thus maintaining constant bandwidth and filter support, and omit generalizations to other graph coarsening schemes, such as Kron-reduction, for simplicity.
\begin{cor} \label{cor3}Consider an undirected, circulant graph $G$ of dimension $N$ and bandwidth $\tilde{B}$, and let ${\bf x}$ be the multiresolution decomposition of graph signal ${\bf x}_W$, which is a $1$-piece polynomial of maximum degree $D\leq 2k-1$, on $G$ via the $j$-level GWT matrix ${\bf W}$. \\
$(i)$ Let ${\bf W}$ be the HGSWT of order $2k$, where the corresponding  low-and high-pass graph filter matrices each have bandwidth $B=k\tilde{B}$, and assume that $B$ is sufficiently small such that $\sum_{n=0}^{l}\frac{B}{2^n}\leq \frac{N}{2^{l+1}}$ at each level $l\leq j-1$. The resulting ${\bf x}={\bf P}{\bf W}{\bf x}_W$ is $K$-sparse, where $K=\frac{N}{2^j}+B (2( j-1)+2^{1-j})$, when $B=2^{j-1} r,\enskip r\in\mathbb{Z}^{+}$.\\
$(ii)$ Let ${\bf W}$ be the HCGSWT  of order $2k$, where the corresponding  low-and high-pass graph filter matrices have bandwidth $T$ and $B=k\tilde{B}$ respectively, such that $B+\sum_{n=1}^{l}\frac{T}{2^n}\leq \frac{N}{2^{l+1}}$ at each level $l\leq j-1$ The resulting ${\bf x}$ is $K$-sparse, where $K=\frac{N}{2^j}+B j+T (j+2^{1-j}-2)$, when $T=2^{j-1} r,\enskip r\in\mathbb{Z}^{+}$.\\
$(iii)$ Let ${\bf W}$ be the HGSWT at $j=0$, with the alternative `minimum' downsampling pattern, which retains only one low-pass component. Then ${\bf x}$ is $K$-sparse with $K=2B$.\end{cor}
\noindent \textit{Proof.} See Appendix $A2$.\\
\\
When $B\in\mathbb{Z}^{+}$, the results of $(i)$ $\&$ $(ii)$ in Cor. \ref{cor3} apply up to a small correction term, which increases with the number of levels $j$. For the multiresolution decomposition of periodic (complex exponential) graph signals with parameter $\alpha=\frac{2\pi k}{N},k\in\mathbb{N}$, we have the maximum sparsity of $K=\frac{N}{2^j}$ at $j$ levels; selectively, for a suitable GWT that retains invertibility under an alternative downsampling pattern, up to $K=1$ can be achieved, following $(iii)$.

\subsection{The Graph FRI-framework}
The traditional FRI-framework is built on the central result that certain classes of non-bandlimited signals with a \textit{finite rate of innovation} can be sampled and perfectly reconstructed using kernels of compact support, which satisfy Strang-Fix conditions (\cite{vetorig}, \cite{vet}); in the discrete domain, this prominently entails that a $K$-sparse signal vector ${\bf x}\in\mathbb{R}^N$ can be perfectly reconstructed from $M\geq 2K$ consecutive sample values $y_n$ of the measurement vector ${\bf y}={\bf F}{\bf x}$, where ${\bf F}\in\mathbb{C}^{N\times N}$ is the DFT matrix, of the form 
\begin{equation}
y_n=\frac{1}{\sqrt{N}}\sum_{k=0}^{K-1}x_{c_k}e^{-i2\pi c_k n/N}=\sum_{k=0}^{K-1}\alpha_k u_k^{n}
\end{equation}
with weights $x_{c_k}$  of ${\bf x}$ at positions ${c_k}$. Here, the locations $u_k=e^{-i2\pi c_k n/N}$ and amplitudes $\alpha_k=x_{c_k}/\sqrt{N}$ are successively recovered using a reconstruction algorithm known as Prony's method \cite{vet}. In particular, the filtering (or acquisition) of a sparse signal with ${\bf F}$ facilitates its exact reconstruction from a dimensionality-reduced version in the Fourier domain.\\
\\
The insight that the graph frequency-ordered GFT basis of an arbirtrary circulant graph $G$ can be expressed as the DFT-matrix subject to a graph-dependent permutation of columns, motivates a direct extension of sparse sampling to the graph-domain, and in the following, we proceed to formulate the FRI-framework for signals on circulant graphs, which we term the Graph FRI-framework (GFRI):  
\begin{thm} \label{samp1} (Graph-FRI) Define  the permuted GFT basis ${\bf U}$ of undirected circulant graph $G$ such that ${\bf U}^H$ is the DFT-matrix. We can sample and perfectly reconstruct a (wavelet-)$K$-sparse graph signal (with multiresolution) ${\bf x}\in\mathbb{C}^N$, on the vertices of circulant $G$ using the dimensionality-reduced GFT representation ${\bf y}={\bf U}^H_M{\bf x}$, ${\bf y}\in\mathbb{C}^M$, where ${\bf U}^H_M$ are the first $M$ rows of ${\bf U}^H$, as long as $M\geq 2K$.\end{thm}
\noindent \textit{Proof.} See Appendix $A.3$. \\
\\
Similarly as in the traditional case, the proof of this theorem is based on the application of Prony's method. In reference to our previous sparsity analysis, we therefore require at least $K<\frac{N}{2}$ for a given $K$-sparse graph signal ${\bf x}$, since $M\geq 2K$, for $N=2^n, n\in\mathbb{Z}^{+}$ by initial assumption.\\
In particular, we further note that since all circulant graphs possess the same sampling basis ${\bf U}^H$, the reduced representation ${\bf y}$ does not directly reveal the underlying graph topology; nevertheless, if the graph is known a priori, the samples (or frequency coefficients) $y(\lambda_i^{\sigma})$ gain a unique spectral interpretation, where $\{\lambda_i^{\sigma}\}_{i=1}^M$ is the partial graph Laplacian spectrum as ordered by the DFT via (graph-based) permutation $\sigma$. In the next step, we thus seek to explicitly derive the graph structure corresponding to ${\bf y}$.
\subsubsection{Graph Coarsening for GFRI}
The problem of downsampling a signal on a graph $G=(V,E)$ along with graph coarsening, as the task of determining the reduced set of vertices and edges of the coarsened graph $\tilde{G}=(\tilde{V},\tilde{E})$, are inherent to GSP theory and represent one of the challenges that the complex data dependencies of graphs impose on traditional signal processing. A variety of approaches have been formulated \cite{mult}, ranging from  spectral graph partitioning, where the largest graph Laplacian eigenvector is used to determine a downsampling pattern, up to graph-specific operations such as for bipartite graphs (\cite{shu}, \cite{ortega2}), which naturally comprise a partitioning into two disjoint sets of nodes. Reconnection may be conducted to satisfy a range of properties, and is an accompanying problem in itself. In the context of multilevel graph wavelet analysis, the properties of most interest here are preservation of circularity and a sparse GWT representation, and as implied by our foregoing discussion, the latter is achieved when the bandwidth of the graph Laplacian is small, i.e. minimal reconnection is conducted. \\
\\ 
In our current set-up, we are interested in identifying the graph structure associated to the dimensionality reduced GFT-representation ${\bf y}$, yet, conversely to general graph coarsening approaches, we need to extract an appropriate downsampling pattern as well as a reconnection strategy from the information given by the spectral coefficients at hand, rather than impose a set of desired properties in the first instance.
The difficulty is posed by the fact that ${\bf y}$ resides in the graph spectral domain and does not directly give rise to a specific downsampling pattern in the vertex domain. 
\\
\\ 
In the traditional FRI-framework \cite{vet}, a given sparse signal can be sampled with a general exponential reproducing kernel $\varphi(t)$, not restricted to the complex exponentials of the DFT as previously shown, where the function $\varphi(t)$ in continuous-time and its shifted versions, is such that it can reproduce exponentials for a proper choice of coefficients $c_{m,n}$
\begin{equation}\label{eq:coeff}\sum_{n\in\mathbb{Z}} c_{m,n} \varphi(t-n)=e^{\alpha_m t}\quad\text{for}\quad \alpha_m\in\mathbb{C},m=0,...,P.\end{equation}
We note that the coefficients $c_{m,n}$ in Eq. (\ref{eq:coeff}) can be expressed as $c_{m,n}=c_{m,0}e^{\alpha_m n}$, where $c_{m,0}=\int_{-\infty}^{\infty}e^{\alpha_m x}\tilde{\varphi}(x) dx$ \cite{vet}. Notably, the functions $\varphi(t)$ and $\tilde{\varphi}(t)$ form a quasi-biorthonormal set, with biorthonormality as a special case (\cite{uri2}, \cite{quasi}).\\
Inspired by this notion of sampling a sparse signal in a multi-layered scheme, we extend the Graph FRI-framework by expressing the reduced GFT-basis ${\bf U}_M^H$ as the product between a fat coefficient matrix ${\bf C}$ and a row-reduced low-pass GWT filter, which can reproduce complex exponential graph signals as per Thms. \ref{esp2} and \ref{esp3}.\\
\\
We proceed to demonstrate the feasibility of this scheme by first proving the existence of such a matrix ${\bf C}$ and its relation to a row-reduced DFT-matrix on the basis of graph e-spline wavelet theory in the following:
\begin{lem} \label{lem11} Let ${\bf U}_M^H$ be the reduced GFT-basis of undirected circulant graph $G$, as defined in Thm. \ref{samp1}, and ${\bf E}_{\vec{\alpha}}\in\mathbb{R}^{N\times N}$ a low-pass graph filter matrix in the e-spline GWT family (see Thms. $3.2$, $3.3$), which can reproduce complex exponential graph signals with parameter $\vec{\alpha}=(\alpha_0,...,\alpha_{M-1})=\left(0,...,\frac{2\pi k}{N},...,\frac{2\pi (M-1)}{N}\right)$. We thus have ${\bf U}_M^H={\bf C}{\bf \Psi}_{\downarrow 2}{\bf E}_{\vec{\alpha}}$, where ${\bf \Psi}_{\downarrow 2}\in\mathbb{R}^{N/2\times N}$ is a binary sampling matrix which retains even-numbered nodes, and ${\bf C}\in\mathbb{C}^{M\times N/2}$ is a coefficient matrix. Further, ${\bf C}=\hat{{\bf C}}\tilde{{\bf U}}_M^H$, where $\hat{{\bf C}}\in\mathbb{C}^{M\times M}$ is diagonal and $\tilde{{\bf U}}^H$ is the DFT matrix of dimension $N/2$.\end{lem}
\begin{proof}
Consider the general complementary graph e-spline wavelet filterbank (Thm. \ref{esp3}) with respective analysis and synthesis matrices \[{\bf W}= \begin{bmatrix} 
 {\bf \Psi}_{\downarrow 2}{\bf H}_{LP_{\alpha}}\\
  {\bf \Phi}_{\downarrow 2}{\bf H}_{HP_{\alpha}}
  \end{bmatrix},\quad \tilde{{\bf W}}=\begin{bmatrix} {\bf \Psi}_{\downarrow 2}\tilde{{\bf H}}_{LP_{\alpha}}\\{\bf \Phi}_{\downarrow 2}\tilde{{\bf H}}_{HP_{\alpha}}\end{bmatrix}\] such that $\tilde{{\bf W}}^T{\bf W}={\bf I}_N$, where the high-pass representer polynomials at both branches possess the same number of vanishing moments, i.e. roots at $z=e^{\pm i\alpha}$ for some $\alpha\in\mathbb{R}$. Let ${\bf x}\in\mathbb{C}^N$ be a complex exponential graph signal of the form
   \[{\bf x}=\begin{bmatrix} e^{i\alpha 0}&
 e^{i\alpha 1}&
 e^{i\alpha 2}&
 \dots &
 e^{i\alpha (N-1)}\end{bmatrix}^T\]
 where $\alpha=-\frac{2\pi k}{N}$, i.e. ${\bf x}^T$ is the $(k+1)$-th row of the (unnormalized) DFT-matrix, and define ${\bf y}={\bf H}_{LP_{\alpha}}{\bf x}=c{\bf x}$ for $c\in\mathbb{R}$ (also an eigenvalue of ${\bf H}_{LP_{\alpha}}$), such that  \[{\bf \Psi}_{\downarrow 2}{\bf H}_{LP_{\alpha}}{\bf x}=c{\bf \Psi}_{\downarrow 2}{\bf x}= c\begin{bmatrix} e^{i\alpha 0}&
 e^{i\alpha 2}&
 e^{i\alpha 4}&
 \dots &
 e^{i\alpha (N-2)}\end{bmatrix}^T={\bf y}(0:2:N-2)={\bf y}_{\downarrow 2}\] which denotes a scalar multiple of the $(k+1)$-th row of the DFT of dimension $N/2$, since
 \[\begin{bmatrix} e^{i\alpha 0}&
 e^{i\alpha 2}&
 e^{i\alpha 4}&
 \dots &
 e^{i\alpha (N-2)}\end{bmatrix}^T=\begin{bmatrix} e^{i(2\alpha) 0}&
 e^{i(2\alpha) 1}&
 e^{i(2\alpha) 2}&
 \dots &
 e^{i(2\alpha) (N/2-1)}\end{bmatrix}^T\]
with $2\alpha=-\frac{2\pi k}{N/2}$. We obtain
 
 \[\begin{bmatrix} ({\bf \Psi}_{\downarrow 2}\tilde{{\bf H}}_{LP_{\alpha}})^T &({\bf \Phi}_{\downarrow 2}\tilde{{\bf H}}_{HP_{\alpha}})^T\end{bmatrix}
\begin{bmatrix} {\bf \Psi}_{\downarrow 2}{\bf H}_{LP_{\alpha}}\\
  {\bf \Phi}_{\downarrow 2}{\bf H}_{HP_{\alpha}}
 \end{bmatrix}
 {\bf x}=\begin{bmatrix} ({\bf \Psi}_{\downarrow 2}\tilde{{\bf H}}_{LP_{\alpha}})^T &({\bf \Phi}_{\downarrow 2}\tilde{{\bf H}}_{HP_{\alpha}})^T\end{bmatrix}\begin{bmatrix}
   {\bf y}_{\downarrow 2}\\
   {\bf 0}_{N/2}\end{bmatrix},\] but since $({\bf \Phi}_{\downarrow 2}\tilde{{\bf H}}_{HP_{\alpha}})^T {\bf 0}_{N/2}={\bf 0}_{N/2}$, neither ${\bf 0}_{N/2}$ nor $({\bf \Phi}_{\downarrow 2}\tilde{{\bf H}}_{HP_{\alpha}})^T $ contribute, and we can thus write
   
   \[\begin{bmatrix} ({\bf \Psi}_{\downarrow 2}\tilde{{\bf H}}_{LP_{\alpha}})^T &({\bf \Phi}_{\downarrow 2}\tilde{{\bf H}}_{HP_{\alpha}})^T\end{bmatrix}\begin{bmatrix}
   {\bf y}_{\downarrow 2}\\
   {\bf 0}_{N/2}\end{bmatrix}=\begin{bmatrix} ({\bf \Psi}_{\downarrow 2}\tilde{{\bf H}}_{LP_{\alpha}})^T\end{bmatrix}\begin{bmatrix}
   {\bf y}_{\downarrow 2}\end{bmatrix}={\bf x}\]
   i.e. linear combinations of the columns of $({\bf \Psi}_{\downarrow 2}\tilde{{\bf H}}_{LP_{\alpha}})^T$ reproduce ${\bf x}$. Rewriting the former, we obtain ${\bf y}_{\downarrow 2}^T{\bf \Psi}_{\downarrow 2}\tilde{{\bf H}}_{LP_{\alpha}}={\bf x}^T$, and reversing the sequence of ${\bf W}$ and $\tilde{{\bf W}}$, and letting ${\bf E}_{\alpha}={\bf H}_{LP_{\alpha}}$, we arrive at
   \[\begin{bmatrix}
   {\bf c}^T
   \end{bmatrix}{\bf \Psi}_{\downarrow 2}{\bf E}_{\alpha}={\bf x}^T\]
with ${\bf c}\in\mathbb{C}^{N/2}$ (${\bf c}=\hat{c} {\bf x}_{\downarrow 2}$ for eigenvalue $\hat{c}$ of $\tilde{{\bf H}}_{LP_{\alpha}}$), ${\bf \Psi}_{\downarrow 2}{\bf E}_{\alpha}\in\mathbb{R}^{N/2\times N}$, and ${\bf x}^T\in\mathbb{C}^N$. By generalizing the RHS to incorporate $M$ stacked complex exponential vectors ${\bf x}$ to form the transposed DFT-matrix $({\bf U}_M^H)^T$, we can similarly show
   \[
   {\bf C}
   {\bf \Psi}_{\downarrow 2}{\bf E}_{\vec{\alpha}}={\bf U}_M^H,\]
   with ${\bf C}=\hat{{\bf C}}\tilde{{\bf U}}_M^H$ and $\vec{\alpha}=(\alpha_1,...,\alpha_M)$. In particular, the matrix $\hat{{\bf C}}$ is diagonal, while $\tilde{{\bf U}}_M^H$ represents the first $M$ rows of the DFT-matrix of dimension $N/2$.\\
	In the case of a bipartite graph, using Thm. \ref{esp2}, we can proceed similarly, and obtain
	$\begin{bmatrix} ({\bf \Psi}_{\downarrow 2}\tilde{{\bf H}}_{LP_{\alpha}})^T\end{bmatrix}\begin{bmatrix}
   {\bf y}_{\downarrow 2}\end{bmatrix}={\bf x}$
	with ${\bf y}_{\downarrow 2}=c {\bf x}_{\downarrow 2}$ for some $c\in\mathbb{R}$, for synthesis low-pass filter $\tilde{{\bf H}}_{LP_{\alpha}}$, which reveals that it reproduces complex exponentials with parameter $\pm\alpha$, just as the analysis low-pass filter, despite not being of the same support. In particular, as we have shown in \cite{splinesw}, the inherent biorthogonality constraints of the wavelet transform impose that the representer polynomials of the bipartite \textit{HGESWT} contain opposing roots respectively for analysis and synthesis, i.e. we have $H_{LP_{\alpha}}(z)=-z^{-1} H_{HP_{\alpha}}(-z)$ and $\tilde{H}_{LP_{\alpha}}(z) =-z\tilde{H}_{HP_{\alpha}}(-z)$. Thus we may interchange the order of synthesis and analysis branch, and obtain
	$\begin{bmatrix} ({\bf \Psi}_{\downarrow 2}{\bf H}_{LP_{\alpha}})^T\end{bmatrix}\begin{bmatrix}
   \hat{c}{\bf x}_{\downarrow 2}\end{bmatrix}={\bf x}$, 
   confirming our previous result that the columns of the adjacency-matrix based, analysis low-pass filter $({\bf \Psi}_{\downarrow 2}{\bf H}_{LP_{\alpha}})^T$ reproduce ${\bf x}$ as a consequence of the generalized Strang-Fix conditions (\cite{esplinewav}, \cite{strang}).\end{proof}

\noindent In order to eventually identify a graph coarsening scheme within our sparse sampling framework, we begin by noting that a sensible downsampling pattern can already be extracted from the previous result, namely the row-reduced low-pass graph filter ${\bf \Psi}_{\downarrow 2}{\bf E}_{\vec{\alpha}}$ samples every other node in keeping with the standard circulant downsampling pattern with respect to $s=1\in S$. \\
A popular graph coarsening scheme known as Kron-reduction \cite{kron} employs Schur complementation based on the given node sampling pattern, where the graph-Laplacian matrix $\tilde{{\bf L}}$ of the coarsened graph is evaluated from the graph Laplacian matrix ${\bf L}$ of initial graph $G$ and set $V_{\alpha}=\{0:2:N-2\}$ of retained nodes:
 \[ \tilde{{\bf L}}={\bf L}(V_{\alpha},V_{\alpha})-{\bf L}(V_{\alpha},V_{\alpha}^{\complement}){\bf L}(V_{\alpha}^{\complement},V_{\alpha}^{\complement})^{-1}{\bf L}(V_{\alpha},V_{\alpha}^{\complement})^T.\]
In particular, it can be shown that for $V_{\alpha}$ as set above and symmetric circulant ${\bf L}$, the resulting coarsened graph will preserve these properties (\cite{kron}, \cite{Ekambaram3}); a drawback is however that it leads to denser graphs, i.e. for a banded circulant matrix, the resulting lower-dimensional Schur complement will be of equal or larger bandwidth which proves destructive in our sparsity-driven filterbank construction and analysis.\\
\\
Alternatively, we propose to conduct sampling in the spectral domain of the graph at hand by leveraging the fact that any graph Laplacian eigenvector ${\bf u}_k$ of $G$ can be interpreted as a graph signal on its vertices with sample value $u_k(i)$ at node $i$, suggesting that an extracted downsampling pattern may be equivalently applied to the eigenbasis of $G$. 
	\begin{lem} \label{lem22} Consider an undirected circulant graph G with generating set $S$, and adjacency matrix ${\bf A}=\frac{1}{N}{\bf U}{\bf \Lambda}{\bf U}^H\in\mathbb{R}^{N\times N}$ with bandwidth $B$, where $\frac{1}{\sqrt{N}}{\bf U}^H$ is the DFT matrix. We downsample by $2$ via the binary matrix 	${\bf \Psi}_{\downarrow 2}\in\mathbb{R}^{N/2\times N}$ on the first $N/2$ rows in ${\bf U}^H$ and eigenvalues ${\bf \Lambda}$, such that $\tilde{{\bf U}}^H={\bf U}^H_{0:N/2-1}{\bf \Psi}^T_{\downarrow 2}$, and $\tilde{{\bf \Lambda}}={\bf \Psi}_{\downarrow 2}{\bf \Lambda}{\bf \Psi}_{\downarrow 2}^T$. The resulting adjacency matrix $\tilde{{\bf A}}=\frac{1}{N/2}\tilde{{\bf U}}\tilde{{\bf \Lambda}}\tilde{{\bf U}}^H\in\mathbb{R}^{N/2\times N\ 2}$ is circulant with the same generating set $S$ as $G$, provided $2B<N/2$.\end{lem}
\noindent \textit{Proof.} See Appendix $A.4$.\\
\\
We observe that the previous result is further reinforced by the fact that adjacency and graph Laplacian matrices of regular graphs possess the same eigenbasis, which, as we will see in a later discussion, is not upheld for e.g. a path graph. In particular, we note that Lemma \ref{lem22} gives rise to an intriguing coarsening strategy for circulant graphs as it preserves both the original connectivity of the graph by retaining the same generating set, as well as the spectral properties given that its eigenvalues and eigenbasis are respectively composed of a subset and subpartition of the original.\\
\\
Following a generalization of Lemma \ref{lem11}, and the preceding discussion, we formulate the graph coarsening scheme to complement our Graph-FRI framework:
\begin{thm} \label{samp2}Given GFT ${\bf y}\in\mathbb{C}^M$, from Thm \ref{samp1}, we determine the coarsened graph $\tilde{G}=(\tilde{V},\tilde{E})$  associated with the dimensionality-reduced graph signal $\tilde{{\bf y}}\in\mathbb{C}^{\tilde{M}}$, via the $j$-level decomposition
\[{\bf y}={\bf U}^H_M{\bf x}={\bf C}\prod_{j=0}^{J-1}({\bf \Psi}_{j\downarrow 2}{\bf E}_{2^j\vec{\alpha}}){\bf x}={\bf C}\tilde{{\bf y}}\]
where ${\bf U}^H_M\in\mathbb{C}^{M\times N}$ is the row-reduced permuted GFT basis (DFT-matrix), ${\bf C}\in\mathbb{C}^{M\times \tilde{M}}$ is a coefficient matrix with $\tilde{M}=\frac{N}{2^J}$ given $M$, ${\bf \Psi}_{j\downarrow 2}\in\mathbb{R}^{N/2^{j+1}\times N/2^j}$ is a binary sampling matrix which retains even-numbered nodes, and ${\bf E}_{2^j\vec{\alpha}}\in\mathbb{R}^{N/2^j\times N/2^j}$ is a (higher-order) graph e-spline low-pass filter on $\tilde{G}_j$, which reproduces complex exponentials at level $j$, with parameter $\vec{\alpha}=(\alpha_0,...,\alpha_{M-1})=\left(0,...,\frac{2\pi (M-1)}{N}\right)$. The associated coarsened graphs $\tilde{G}_j$ at levels $j\leq J$ can be determined following two different schemes:\\
$(i)$ Perform Kron-reduction at each level $j\leq J$ using the pattern ${\bf V}_{\alpha}$ in ${\bf \Psi}_{j\downarrow 2}$ to obtain ${\bf L}_j$\\
$(ii)$ Define eigenbasis $(\tilde{{\bf U}}_j,\tilde{{\bf \Lambda}}_j)\in\mathbb{C}^{N/2^j \times N/2^j}$ at each level $j \leq J$ through the application of ${\bf \Psi}_{j-1\downarrow 2}$ on $(\tilde{{\bf U}}_{j-1},\tilde{{\bf \Lambda}}_{j-1})$ (see Lemma \ref{lem22}). The coarse graph $\tilde{G}_j$ for graph signal $\tilde{{\bf y}}_j=\prod_{k=0}^{j-1}({\bf \Psi}_{k \downarrow 2}{\bf E}_{2^k\vec{\alpha}}){\bf x}$, has adjacency matrix \[{\bf A}_j=(2^j/N)\tilde{{\bf U}}_j \tilde{{\bf \Lambda}}_j\tilde{{\bf U}}_j^H\]
which preserves the generating set $S$ of $G$ for a sufficiently small bandwidth.\end{thm}
\noindent Consequentially, the edge set of the coarsened graph associated with the GFRI-framework is not unique, and we have explored two possible approaches which satisfy the connectivity constraints of symmetry and circularity. Kron-reduction preserves basic graph characteristics, yet, while taking into account the entire graph adjacency relations in the computation of the coarsened version, it provides little general intuition on the topology of the latter. In contrast, the alternative spectral reduction technique, is shown to preserve the original graph connectivity by retaining its generating set (see for example Fig. \ref{fig:gcc}), thereby simultaneoulsy alleviating the issue of an increasing bandwidth. \\
We summarize the graph sampling framework, further illustrated in Fig. \ref{fig:ssc}, as the filtering  of a sparse graph signal ${\bf x}$ on $G$ with a graph e-spline low-pass filter followed by dimensionality reduction, and giving rise to signal $\tilde{{\bf y}}$ on coarsened $\tilde{G}$, which is subsequently transformed into the further dimensionality-reduced, scaled spectral graph domain, resulting in the representation ${\bf y}$. Graph signals ${\bf x}_W$ with a sparse multiresolution representation via a GWT can be similarly sampled following an initial sparsification step.
The graph filter(s) and transformation ${\bf C}$, which contains graph filter eigenvalues, within the derived decomposition further facilitate a spectral characterization of ${\bf y}$ that depends directly on and is unique for the graph at hand.
\\
\begin{figure}[htb]\centering
\begin{minipage}{6in}
  \centering
  \raisebox{-0.5\height}{\includegraphics[height=1.5in]{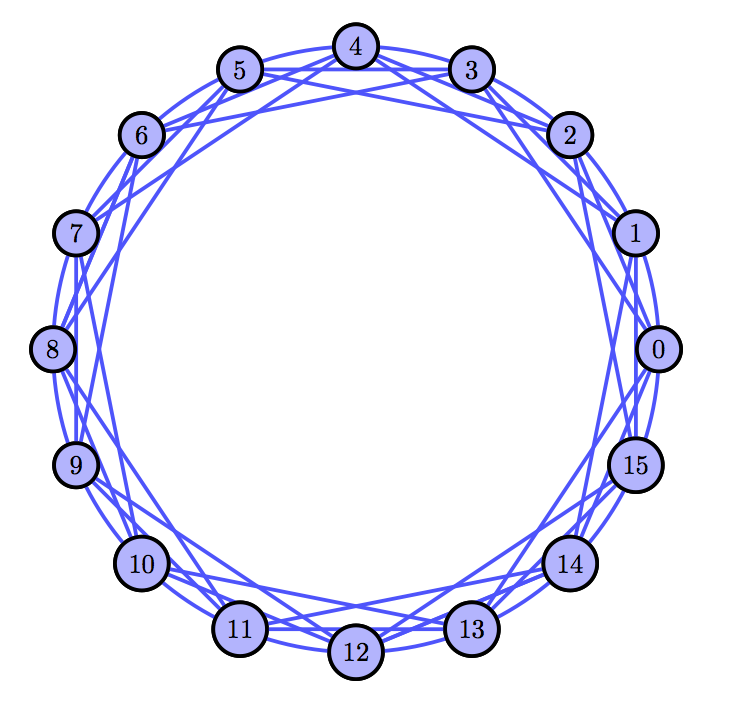}}
  \hspace*{.06in}
 \raisebox{-0.5\height}{\includegraphics[height=0.6in]{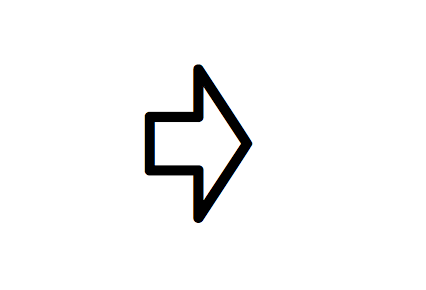}}
 \hspace*{.06in}
  \raisebox{-0.5\height}{\includegraphics[height=1.4in]{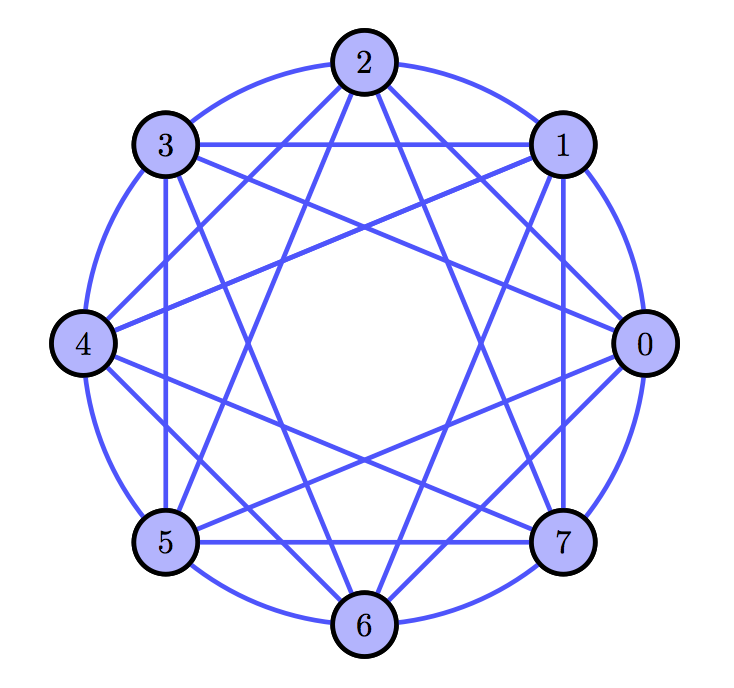}}
\end{minipage}
\caption{Graph Coarsening for a Circulant Graph with $S=\{1,2,3\}$.}
\label{fig:gcc}
\end{figure}

\begin{figure}[htb]\centering
 \centerline{\includegraphics[scale=0.345]{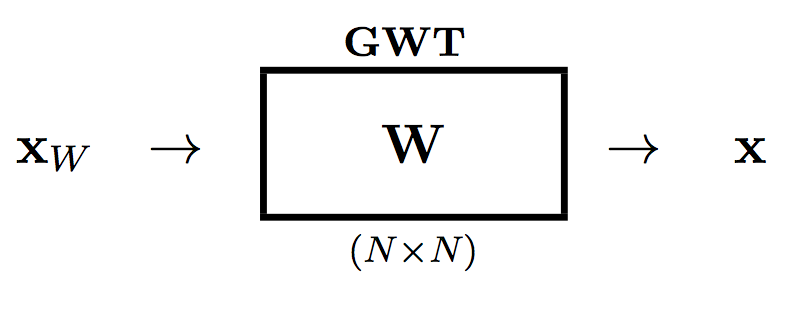}}
 \centerline{\includegraphics[scale=0.45]{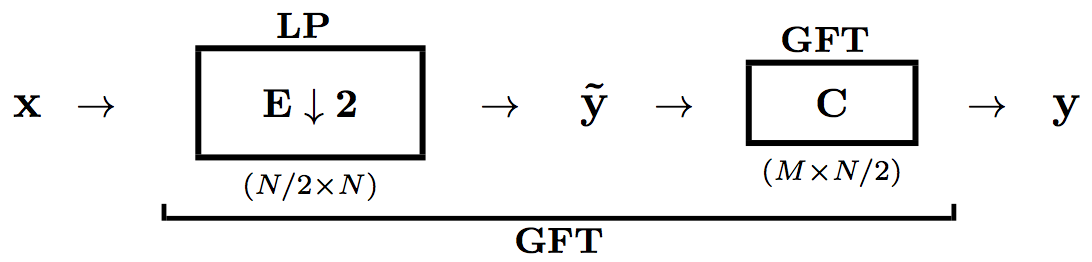}}
\caption{Sampling Scheme with Preceding Sparsification Step and One Level of Coarsening.}
\label{fig:ssc}
\end{figure}
The matrix ${\bf E}_{2^j\vec{\alpha}}$, representing a higher-order, vertex-localized graph e-spline wavelet low-pass filter parameterised by $\vec{\alpha}2^j$ at level $j$, is of the form of the low-pass filter designed in Eqs. (\ref{eq:lp2}) (Thm. \ref{esp2}) or (\ref{eq:lp3}) (Thm. \ref{esp3}), depending on whether the graph at hand is bipartite or not; we thus implicitly assume that ${\bf E}_{2^j\vec{\alpha}}$ can reproduce complex exponential graph signals, with $\vec{\alpha}$ as specified. As the filter construction in either case is based on the combination/convolution of different graph e-spline basis functions, it should be noted that the resulting higher-order function may for certain scenarios contain opposing roots, which in the non-bipartite case would violate B\'{e}zout's Thm. \cite{esplinewav}, and thus the necessary biorthogonality condition for filterbank construction. While for complementary graph wavelet filterbanks, it is more intricate to deduce generalized conditions on when exactly this occurs, we can further specify conditions under which the \textit{HGESWT} in the bipartite graph case loses reproduction properties and/or ceases to be invertible, as formulated in Cor. \ref{cor1} \footnote{In particular, this amounts to showing when $\tilde{d}_{\alpha_j}=\sum_{k\in2\mathbb{Z}+1}2d_k\cos\left(\frac{2 \pi j k}{N}\right)=-\tilde{d}_{\alpha_t}=\sum_{k\in2\mathbb{Z}+1}2d_k\cos\left(\frac{2 \pi k}{N}\left(t\pm \frac{N}{2}\right)\right)$ can occur, i.e. for which $\alpha_j=\frac{2\pi j}{N}$, $\alpha_t=\frac{2\pi t}{N}$, weights $d_k$ and graph connectivity $k\in2\mathbb{Z}+1$, the scheme ceases to be valid.}.
\subsection{Special Cases and Discussion}
Graph spline wavelet theory and ensuing graph-dependent transform properties in particular, as discussed in Sect. $3$, predetermine the extent to which one can characterize a coarsened graph associated with the sampled graph signal $\tilde{{\bf y}}$, or, in other words, the degree of feasible dimensionality reduction ($j$-level decomposition) for the initial graph within the proposed GFRI-framework.\\
Further to Thm \ref{samp2}, we hence proceed to specify restrictions on the number of samples $M$ and levels $J$ to ensure the reproduction via low-pass graph filter ${\bf E}_{\vec{\alpha}}$. Certain rows of the DFT-matrix cannot be reproduced using a real-valued symmetric graph filter, as alluded to in (Sect. $3$, \cite{splinesw}); for instance, parameter $\alpha=\pm \pi/2$ induces the opposing factors $(1-iz)(1+iz)$ of complex conjugates in the representer polynomial of a parameterized graph Laplacian filter, which violates B\'{e}zouts equality for complementary filterbank construction, as well as prevents the reproduction of the corresponding complex exponential in the DFT via a \textit{HGESWT}-based low-pass graph filter. It follows that we cannot reproduce consecutive rows of the DFT-matrix beyond its $N/4$-th row. \\
We therefore need to ensure within the multiresolution analysis that\[\alpha 2^{j-1}=\frac{2 \pi k 2^{j-1}}{N}< \frac{2 \pi (N/4)}{N},\quad\forall j\leq J\] 
such that $k<\frac{N}{2^{J+1}}$ or $J<\log_2\left(\frac{N}{k}\right)-1$, where $J$ is the total number of levels; in other words, we can approximate the DFT-matrix up to its $M=k+1$-th row for a certain number of levels, with parameters $\tilde{M}=\frac{N}{2^J}$ and $M=\frac{N}{2^{J+1}}=\frac{\tilde{M}}{2}$ of Thm \ref{samp2}. This coincides with the biorthogonality constraint for traditional e-spline wavelets outlined in \cite{espline}, \cite{esplinewav} which ensures that the corresponding filters do not contain opposing roots, whereby given distinct $\gamma,\gamma'\in\vec{\gamma}$, $2^j(\gamma-\gamma')=i\pi (2k+1)$ must not be satisfied for some $k\in\mathbb{Z}$ at level $j\leq j_0-1$. 
Depending on the signal set, we may not always satisfy $K=\frac{N}{2^{J+2}}$ exactly, in which case we simply require $2(2K)\leq 2M\leq\tilde{M}$ and generalize the formulae accordingly. \\
Since the graph e-spline filter functions of Thm. \ref{esp2} give rise to the traditional e-spline when the graph at hand is a simple cycle, the aforementioned condition with $\gamma=\pm i\alpha$ is sufficient in that case; nevertheless as a result of the complex connectivity of circulant graphs and thereby arising eigenvalue multiplicities, we need to impose further restrictions on $\vec{\alpha}$ for all other cases.\\
\\
For an unweighted bipartite graph with consecutive, odd elements $s_k\in S$ and even generating set cardinality $|S|$, and as a particular case of Cor. \ref{cor1}, we additionally observe at $\alpha=\pm\pi/4$ for e-degree $\tilde{d}_{\alpha}=\sum_{k\in\mathbb{Z}} 2 \cos(\alpha (2k+1))=0$, since $\cos(\alpha (2k+1))=\frac{\sqrt{2}}{2},\quad k=0,3,4,7,8,...$, and $\cos(\alpha (2k+1))=-\frac{\sqrt{2}}{2},\quad k=1,2,5,6,...$. Thus, similarly as before, the associated graph filter polynomials contain opposing roots and we can only approximate consecutive rows of the DFT-matrix up to at most the $N/8$-th row. This translates into the constraints $k<\frac{N}{2^{J+2}}$ or $J<\log_2\left(\frac{N}{k}\right)-2$, and $M=\frac{N}{2^{J+2}}=\frac{\tilde{M}}{4}$. Here, $M$ may be subjected to further reduction, as a consequence of increasing eigenvalue multiplicities at $\tilde{d}_{\alpha}=\gamma_i=0$ for different graph-connectivities.\\
\\
In general, if $G$ is circulant and bipartite, we need to ensure that no parameters $\alpha_i,\alpha_j\in\vec{\alpha}$ satisfy $\tilde{d}_{\alpha_i}= -\tilde{d}_{\alpha_j}$ at all levels in order to preserve the invertibility property of Thm \ref{esp2}, and thus consider consecutive frequencies (or consecutive rows of the DFT), only up to some cut-off frequency with $\alpha_k=\frac{2\pi k}{N}$ at position $k+1$, such that for $i,j\leq k$, $\tilde{d}_{\alpha_i}\neq -\tilde{d}_{\alpha_j}$. 
As per Cor. \ref{cor2}, $\tilde{d}_{\alpha_i}= -\tilde{d}_{\alpha_j}$ is satisfied when $j=(i+N/2)_N$ for frequency location parameters $i,j$, which, despite the previously derived constraint $i,j< N/4$, may ensue for some $i,j$ from large eigenvalue multiplicities at $0$ (associated with higher graph connectivity). An example is given by the normalized adjacency matrix of the unweighted complete (circulant) bipartite graph with bipartite sets of equal size $N/2$, whose eigenvalues are $\gamma_{max/min}=\pm1$ of respective multiplicity $m=1$, and $\gamma_i=0$ of multiplicity $N-2$.\\
These are necessary conditions for the existence of a suitable low-pass filter via the \textit{HGESWT}, we note, however, that the set of special cases presented here is not exhaustive. 
It can be further stated that the condition $j\neq(i+N/2)_N$ also needs to be satisfied for non-bipartite circulant graphs in the \textit{HCGESWT}, which equivalently follows from the traditional biorthogonality constraints as well as from a special case of the presented graph spline wavelet transform designs.\footnote{This can be easily demonstrated in a generalization of the proof of Thm. \ref{esp2} (in Appendix $A.2$, \cite{splinesw}), where high-pass filter ${\bf H}_{HP_{\vec{\alpha}}}$ is maintained and low-pass filter ${\bf H}_{LP_{\vec{\alpha}}}$ is generalized to the form of Eq. (\ref{eq:lp3}), with fixed downsampling pattern with respect to $s=1\in S$.}
\\
\\
In the following, the derived framework is further illustrated through a sample scenario:\\
\\
{\bf Example:} Consider the bipartite circulant graph of dimension $N=|V|=128$ with generating set $S=\{1,3,5\}$ (as illustrated in Sect. $3$, Fig. \ref{fig:gsignal}) and GFT (DFT) basis ${\bf U}^H$, and define a set of $K$-sparse signals $X=\{{\bf x}_i\}_i$, with $K=3$, on its vertices. Here, the elements of $X$ may characterize i.a. (piecewise) smooth signals, which have been sparsified via a suitable multi-level GWT.
According to Thm. \ref{samp1}, one can perfectly recover each ${\bf x}_i\in\mathbb{R}^N$ from its associated dimensionality-reduced spectral signal representation ${\bf U}^H_M{\bf x}_i={\bf y}_i\in\mathbb{R}^M$ of minimum dimension $M=2K=6$. \\
As such, we require $M-1=k<\frac{N}{2^{J+1}}$, and $\tilde{M}\geq 2M=12$ for the dimension of coefficient matrix (spectral transformation) $\tilde{{\bf C}}\in\mathbb{C}^{M\times \tilde{M}}$, and obtain $J=3 <\log_2\left(\frac{N}{k}\right)-1$ for the maximum number of decomposition levels in the graph coarsening scheme of Thm. \ref{samp2}, with $\tilde{M}=\frac{N}{2^J}=16$. \\
At each level $0\leq j\leq 2$, we proceed to establish a suitably coarsened bipartite circulant graph $\tilde{G}_j$ with adjacency matrix ${\bf A}_j$, by retaining the same generating set $S$ as $G$ (as per Thm. \ref{samp2} $(ii)$), and define the \textit{HGESWT} (Thm. \ref{esp2}) on each, with low-pass filter ${\bf E}_{2^j \vec{\alpha}}=\prod_{n=0}^5\frac{1}{2}\left(\frac{\tilde{d}_{2^j\alpha_n}}{d}{\bf I}_{N/2^j}+\frac{{\bf A}_j}{d}\right)$ parameterized by $\{\alpha_n=\frac{2\pi n}{N}\}_{n=0}^5$ to reproduce complex exponentials (i.e. the first $M=6$ rows of the DFT).\\
Upon numerical inspection, we observe that $\tilde{d}_{2^j\alpha_k}\neq -\tilde{d}_{2^j\alpha_l}$ generally holds for any two e-degree parameters, with all $\tilde{d}_{2^j\alpha_k}\neq 0$, while the \textit{HGESWT} (here, for $k=1$) is invertible at all levels $0\leq j\leq 2$. \\
This gives rise to the set of sampled graph signals $\tilde{{\bf y}}_i^j=\prod_{k=0}^{j-1}({\bf \Psi}_{k \downarrow 2}{\bf E}_{2^k\vec{\alpha}}){\bf x}_i$ on each $\tilde{G}_j$; here, $\tilde{{\bf y}}=\prod_{k=0}^{2}({\bf \Psi}_{k \downarrow 2}{\bf E}_{2^k\vec{\alpha}}){\bf x}_i\in\mathbb{R}^{\tilde{M}}$ is defined on the coarsened circulant graph $\tilde{G}$ of minimum dimension $\tilde{M}=16$, characterized by the same generating set $S$ as $G$.

\subsection{Extensions to Path Graphs}
The path graph, which corresponds to a simple cycle without the periodic extension, bears similar properties to its circulant counterpart; it is bipartite and its graph Laplacian eigenvectors can be represented as the basis vectors of the DCT-II matrix \cite{strang2} such that ${\bf U}^H={\bf Q}$ is the DCT-III matrix, with entries $Q_{m,n}=c(m)\sqrt{\frac{2}{N}}\cos\left(\frac{\pi m (n+0.5)}{N}\right)$, for $0\leq m,n\leq N-1$, and constants $c(0)=\frac{1}{\sqrt{2}}$ and $c(m)=1$ for $m\geq 1$, with corresponding distinct eigenvalues $\lambda_m=2-2\cos\left(\frac{\pi m}{N}\right),\quad m\in\{ 0, 1,...,N-1\}$. According to \cite{drag}, a $K$-sparse signal sampled with the DCT matrix can be perfectly reconstructed via a variation of Prony's method using at least $4K$ of its consecutive sample values which gives rise to a specialized extension of the Graph FRI framework\footnote{In \cite{drag}, the DCT matrix is given as an example of a larger class of invertible sampling bases of the form ${\bf Q}={\bf \Lambda}{\bf V}{\bf S}$, with diagonal ${\bf \Lambda}\in\mathbb{C}^{N\times N}$, Vandermonde matrix ${\bf V}\in\mathbb{C}^{N\times M}$ with $\lbrack{\bf V}\rbrack_{n,m}=p_m^n$ and distinct $p_m$, and ${\bf S}\in\mathbb{C}^{M\times N}$, whose columns are at most $D$-sparse. A $K$-sparse signal ${\bf x}$ can be perfectly recovered from $2DK$ consecutive entries of ${\bf y}={\bf Q}{\bf x}$ using Prony's method (Prop. $4$, \cite{drag}), which allows a further generalization to graphs whose GFT basis is of that form.}:
\begin{thm} (Graph-FRI for paths) Let ${\bf x}\in\mathbb{C}^N$ be a K-sparse graph signal defined on the vertices of an undirected and unweighted path graph $G$, whose GFT basis is expressed such that ${\bf U}^H$ is the DCT-matrix ${\bf Q}$. We can sample and perfectly reconstruct ${\bf x}$ on $G$ using the dimensionality-reduced GFT-representation ${\bf y}={\bf U}^H_M {\bf x}\in\mathbb{C}^M$, where ${\bf U}^H_M$ corresponds to the first $M$ rows of ${\bf U}^H$, provided $M\geq 4K$.\end{thm}
\noindent Furthermore, we note that the graph Laplacian matrix of a path is circulant up to its first and last row, thus incorporating $2$ vanishing moments with a graph-inherent border effect of length $T=2$; powers of the graph Laplacian similarly inherit $2k$ vanishing moments with a border effect of $T=2k$. By considering the symmetric normalized adjacency matrix ${\bf A}_{n}={\bf D}^{-1/2}{\bf A}{\bf D}^{-1/2}$, the previously derived graph wavelet construction of Thm. \ref{esp1} can be generalized for the path graph, even though not a regular graph, with the minor restriction that the normalized graph Laplacian ${\bf L}_{norm}={\bf I}_N-{\bf A}_n$ gains an increased border effect of $T=2(k+1)$. More generally, the proof of Thm. \ref{esp1} can be extended to all undirected connected graphs \cite{splinesw}, given that the eigenvalues $\gamma_i$ of ${\bf A}_n$ continue to satisfy $|\gamma_i|\leq 1$. We thus state the graph wavelet transform as a special case of the \textit{HGSWT} in Thm \ref{esp1}:
\begin{cor} Given the undirected path graph $G$ with normalized adjacency matrix ${\bf A}_{n}$, we define the HGSWT, composed of the low-and high-pass filters:
\begin{equation}{\bf H}_{LP}=\frac{1}{2^k}({\bf I}_N+{\bf A}_n)^k\end{equation}
\begin{equation}{\bf H}_{HP}=\frac{1}{2^k}({\bf I}_N-{\bf A}_n)^k\end{equation}
whose associated high-pass representer polynomial $H_{HP}(z)$ annihilates polynomial graph signals up to degree $2k-1$, subject to a graph-border effect of $T=2(k+1)$ non-zeros. This filterbank is invertible for any downsampling pattern, as long as at least one node retains the low-pass component, while the complementary set of nodes retains the high-pass components.\end{cor}
\noindent Here, the nullspace of ${\bf L}_{norm}={\bf I}_N-{\bf A}_n$ does not contain the all-constant vector, but rather the variation ${\bf D}^{1/2}{\bf 1}_N$, with ${\bf 1}_N$ as the $N$-dimensional vector of ones. \\
Similar extensions pertain to the e-graph-spline transform in Thm. \ref{esp2}, whose filters we can generalize to be of the form ${\bf H}=\prod_j\frac{1}{2^k}(\lambda_j{\bf I}_N \pm {\bf A}_n)^k$ with respect to the eigenvalues $\{\lambda_j\}_j$ of ${\bf A}_n$, in the spirit of the eigenspace-shift discussed in Sect. $3$; these graph filters are shift-invariant with respect to the normalized adjacency matrix since they are formed by polynomials in ${\bf A}_n$. However, we are less interested in these results, except in order to achieve a sparse multiresolution representation, given that the DCT does not give rise to an equivalently intuitive decomposition scheme as the DFT for sampling-based graph coarsening.\\
\\
At last, it should be noted that generalized graph coarsening of a path graph, within a multilevel graph wavelet analysis, can i.a. be conducted via the Kron-reduction of the graph Laplacian matrix, where every other node is sampled, resulting in a weighted path graph with universal weight $1/2$ \cite{mult}, as well as via a spectral sampling scheme similar to Lemma \ref{lem22}, as a consequence of its near-circulant structure.

\section{Generalized $\&$ Multidimensional Sparse Sampling}
In order to apply the presented sampling framework to sparse signals defined on arbitrary graphs, one can make use of a variety of approximation schemes, which facilitate the interpretation of circulant graphs as building blocks for the former. Given the adjacency matrix ${\bf A}$ of a general graph $G$, we propose to conduct $(i)$ nearest circulant approximation of ${\bf A}$ by $\tilde{{\bf A}}$, or alternatively $(ii)$ the (approximate) graph product decomposition ${\bf A}\approx {\bf A}_1\diamond {\bf A}_2$ into circulant graph factors ${\bf A}_i$. \\
In the former case, this entails the projection of ${\bf A}$ onto the subspace of circulant matrices ${\bf C}_N$, spanned by circulant permutation matrices ${\bf \Pi}^i,\enskip i=0,...,N-1$, with ${\bf \Pi}$ defined through first row $\lbrack 0\enskip 1\enskip 0...\rbrack$, and given by $\tilde{{\bf A}}=\sum_{i=0}^{N-1}\frac{1}{N}\langle {\bf A},{\bf \Pi}^i\rangle_F{\bf \Pi}^i$, which performs an averaging over diagonals via the Frobenius inner product $\frac{1}{N}\langle {\bf A},{\bf \Pi}^i\rangle_F=\frac{1}{N}tr({\bf A}^T {\bf \Pi}^i)$ \cite{chan}. 
Here, ${\bf A}$ may be subjected to a prior node relabelling for bandwidth minimization (using for instance the RCM-algorithm \cite{rcm}) so as to reduce the approximation effect of introducing additional complementary edges in $\tilde{{\bf A}}$ and hence significantly alter the graph, as well as prior partitioning, e.g. when the graph at hand features distinct communities that may be analyzed separately. The set of (wavelet-)sparse signals defined on $G$ can then be sampled with respect to its graph approximation $\tilde{G}$. \\
In the latter case, we introduce an additional degree of dimensionality reduction through the graph product operation, which can be successfully leveraged for suitably defined multi-dimensional sparse signals consisting of sparse tensor factors.\\
\\
Analogously to our prior investigation of multi-dimensional graph wavelet analysis, outlined in \cite{splinesw}, we can thus generalize the GFRI sampling framework to arbitrary graphs by resorting to graph product decomposition (\cite{handbook}, \cite{bigdata}); for completeness, we briefly review the main aspects.
\subsection{Sampling on Product Graphs}
The product $\diamond$ of two graphs $G_1=(V(G_1),E(G_1))$ and $G_2=(V(G_2),E(G_2))$, also referred to as factors, with respective adjacency matrices ${\bf A}_1\in\mathbb{R}^{N_1\times N_1}$ and ${\bf A}_2\in\mathbb{R}^{N_2\times N_2}$, gives rise to a new graph $G_{\diamond}$ with vertex set $V(G)=V(G_1)\times V(G_2)$ as the Cartesian product of the former, and edge set $E(G)$ which is formed according to adjacency rules of the respective product operation, resulting in adjacency matrix ${\bf A}_{\diamond}\in\mathbb{R}^{N_1N_2\times N_1N_2}$ \cite{handbook}. We identify four main graph products of interest:
\begin{itemize}
\item Kronecker product $G_1\otimes G_2$: ${\bf A}_{\otimes}={\bf A}_1\otimes {\bf A}_2$
\item Cartesian product $G_1\times G_2$: ${\bf A}_{\times}={\bf A}_1\times {\bf A}_2={\bf A}_1\otimes {\bf I}_{N_2} +{\bf I}_{N_1}\otimes {\bf A}_2$
\item Strong product $G_1\boxtimes G_2$: ${\bf A}_{\boxtimes}={\bf A}_1\boxtimes {\bf A}_2={\bf A}_{\otimes}+{\bf A}_{\times}$
\item Lexicographic product $G_1\lbrack G_2\rbrack$: ${\bf A}_{\lbrack\enskip \rbrack}={\bf A}_1\lbrack{\bf A}_2\rbrack={\bf A}_1\otimes {\bf J}_{N_2} +{\bf I}_{N_1}\otimes {\bf A}_2$
\end{itemize}
where ${\bf J}_{N_2}={\bf 1}_{N_2}{\bf 1}_{N_2}^T$. Here, the lexicographic product is a variation of the Cartesian product. When $G_i$ (and hence $G_{\diamond}$) are regular as well as connected, both adjacency matrix ${\bf A}_{\diamond}$ and graph Laplacian matrix ${\bf L}_{\diamond}$ possess the same eigenbasis ${\bf U}={\bf U}_1\otimes {\bf U}_2$ for ${\bf A}_i={\bf U}_i{\bf \Lambda}_i{\bf U}_i^{H}$ on $G_i$, with graph adjacency eigenvalues ${\bf \Lambda}_{\diamond}={\bf \Lambda}_1\diamond {\bf \Lambda}_2$, except under the lexicographic product \cite{laplaceproduct1}.\\
\\
Inspired by the consideration of graph products in other contexts as a means to model higher-dimensional data and/or facilitate efficient implementation, and a preliminary consideration in \cite{bigdata} for GSP, we proceed to interpret a graph signal residing on the vertices of a product graph as follows \cite{splinesw}:
\begin{defe} Any graph signal ${\bf x}\in\mathbb{R}^N$, with $N=N_1 N_2$, can be decomposed as ${\bf x}=\sum_{s=1}^k{\bf x}_{s,1}\otimes {\bf x}_{s,2}=vec_r\{\sum_{s=1}^k{\bf x}_{s,1}{\bf x}_{s,2}^T\}$, where $vec_r\{\}$ indicates the row-stacking operation, or, equivalently, $\sum_{s=1}^k{\bf x}_{s,1}{\bf x}_{s,2}^T$ has rank $k$ with ${\bf x}_{s,i}\in\mathbb{R}^{N_i}$. For ${\bf x}$ residing on the vertices of an arbitrary undirected graph $G$, which admits the graph product decomposition of type $\diamond$, such that $G_{\diamond}=G_1 \diamond G_2$ and $|V(G_i)|=N_i$, we can redefine and process ${\bf x}$ as the graph signal tensor factors ${\bf x}_{s,i}$ on $G_i$.\end{defe}
Given the graph product decomposition $G_{\diamond}=G_1\diamond G_2$, which may be exact or approximate such that the $G_i$ are undirected, circulant and connected with $s=1\in S_i, i=1,2$, the tensor factors of a given graph signal ${\bf x}$ on $G$ can thus be processed with respect to its inherent circulant substructures; for simplicity, we consider ${\bf x}={\bf x}_1\otimes {\bf x}_2$ (of rank $k=1$) for the remainder of this discussion.\\
\\
The multi-dimensional $K=K_1 K_2$-sparse graph signal ${\bf x}$ with $K_i$-sparse tensor factors ${\bf x}_i$, residing on the vertices of an arbitrary graph $G$, can be sampled and perfectly recovered 
based on dimensionality-reduced GFT-representations of ${\bf x}_i$ on the approximate or exact graph product decomposition of $G$ into circulant factors $G_i$. By applying the GFRI framework (Thms. \ref{samp1} $\&$ \ref{samp2}) on each graph component individually, one can perfectly recover ${\bf x}_i$ (using Prony's method) from spectral representations ${\bf y}_i={\bf U}^H_{M_i}{\bf x}_i$, of dimension $M_i\geq 2K_i$, with ${\bf U}^H_{M_i}$ denoting the first $M_i$ rows of the permuted GFT (DFT) matrix of dimension $N_i\times N_i$, $i=1,2$. 
In particular, for all but the lexicographic product, we have \[({\bf \Phi}_{M_1}\otimes {\bf \Phi}_{M_2}){\bf U}^H{\bf x}=({\bf U}^H_{M_1}\otimes{\bf U}^H_{M_2})({\bf x}_1\otimes {\bf x}_2)={\bf y}_1\otimes {\bf y}_2={\bf y},\] where ${\bf y}\in\mathbb{C}^{M}$ with $M\geq 4 K$ and ${\bf \Phi}_{M_i}\in\mathbb{R}^{M_i\times N_i}$ sample the first $M_i$ rows, or alternatively, ${\bf y}=({\bf C}_1\tilde{{\bf y}}_1)\otimes ({\bf C}_2\tilde{{\bf y}}_2)$ for graph-filtered representations $\tilde{{\bf y}}_i$ and graph spectral transformation matrices ${\bf C}_i$. Otherwise, an entirely separate processing of ${\bf x}_i$ is conducted on the individual graph Laplacian eigenbases of $G_i$. The coarsened circulant graphs $\tilde{G}_i$ associated with representations ${\tilde{\bf y}}_i$ can be recombined under the same graph product operation to form the new coarse graph $\tilde{G}=\tilde{G}_1\diamond\tilde{G}_2$ on the vertices of which the signal $\tilde{{\bf y}}_1\otimes\tilde{{\bf y}}_2$ is redefined. 
\subsection{Exact vs Approximate Graph Products}
If the decomposition $G_{\diamond}=G_1\diamond G_2$ into circulant $G_i$ is exact and known, we can perfectly recover a multidimensional-sparse signal ${\bf x}$ on $G$ by performing graph operations in smaller dimensions, while only requiring the storage of lower dimensional spectral representations ${\bf y}_i$; this advantage is particularly evident in the case of the lexicographic product which is closed under circulant graphs \cite{lex}. Here, one may on the one side apply the GFRI framework directly on the original $G$ with known lexicographic decomposition, conditional upon ${\bf x}$ being sufficiently sparse (or smooth on $G$), yet a decomposition into lower-dimensional circulants increases efficiency, while preserving the scheme.\\
\begin{figure}[htb]\centering
\begin{minipage}{6in}
  \centering
  \raisebox{-0.5\height}{\includegraphics[height=1in]{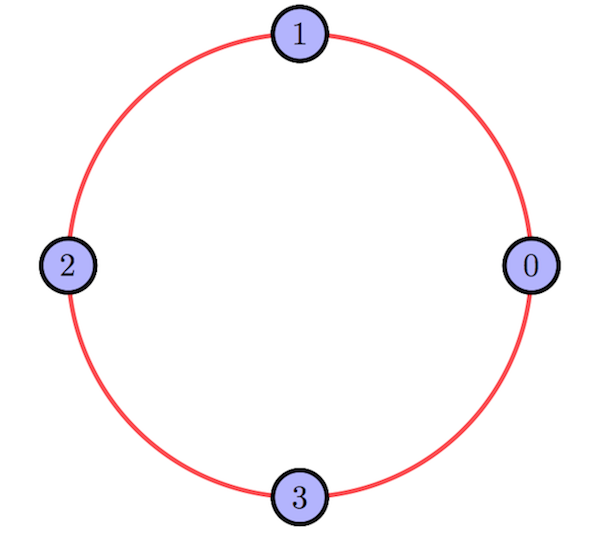}}
  \hspace*{.02in}
  ${\bf \times}$
  \hspace*{.06in}
 \raisebox{-0.5\height}{\includegraphics[height=1in]{r11.png}}
 \hspace*{.03in}
 {\bf =}
 \hspace*{.03in}
  \raisebox{-0.5\height}{\includegraphics[height=2.1in]{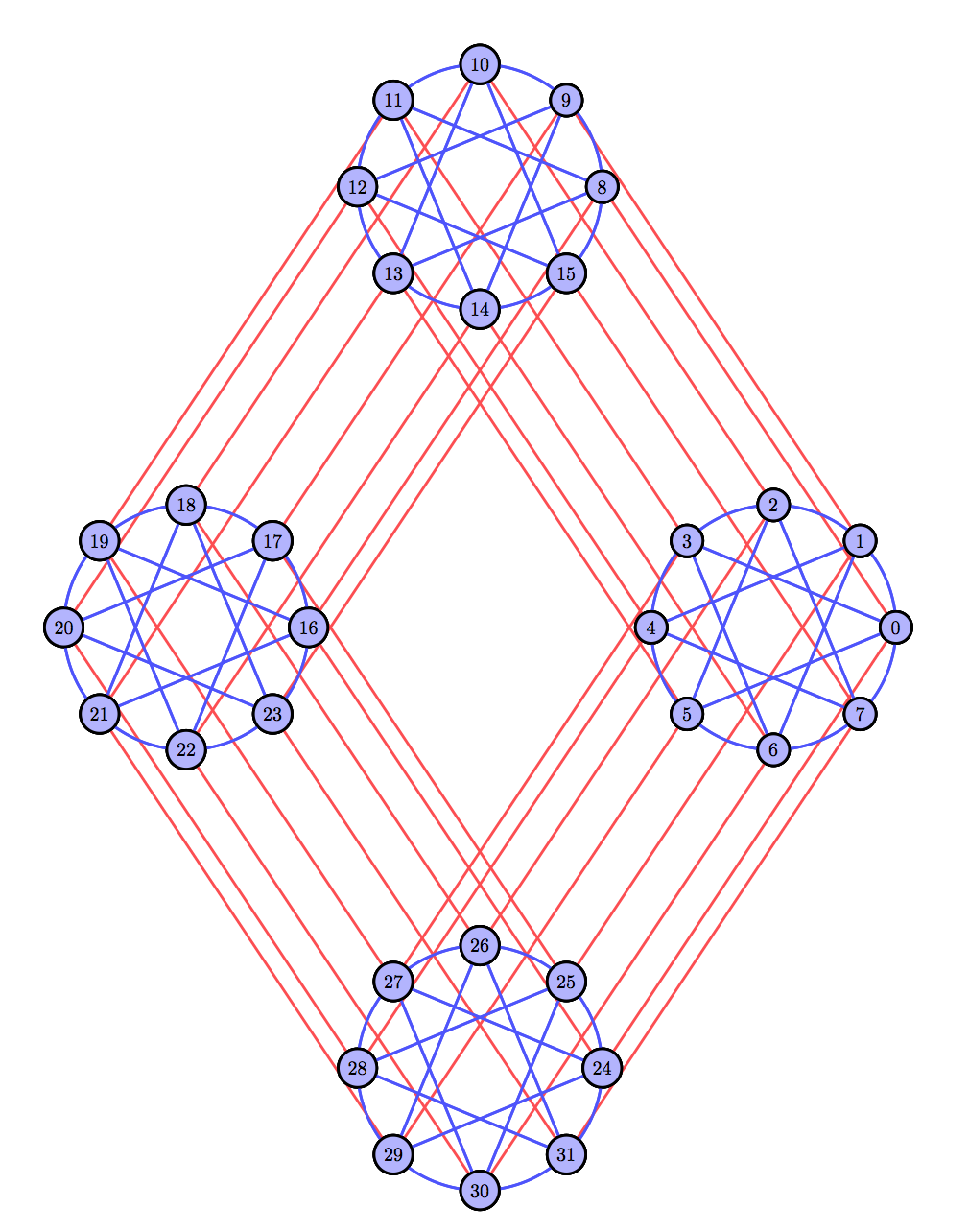}}
\end{minipage}
\caption{Graph Cartesian Product of two unweighted, bipartite circulant graphs.}
\label{fig:gp}
\end{figure}
\\
The study of graph products has further revealed the unweighted lattice graph as the Cartesian product of two unweighted path graphs \cite{handbook}, facilitating the generalization of both the sparse sampling and wavelet analysis framework to lattice graphs, as well as to more general graph products of path and circulant graphs (see Figs. \ref{fig:gp}, \ref{fig:gpp} for examples). 
The vanishing moment property of the graph Laplacian of the path graph facilitates a multidimensional wavelet analysis of multidimensional (piecewise) smooth signals on lattice graphs, thereby revealing an interesting relation to the interpretation of the graph Laplacian as a differential operator.
In particular, the graph Laplacian for lattice graphs provides the stencil approximation of the second order differential operator up to a sign \cite{Ekambaram3}; coincidentally, the unweighted lattice graph, as the graph product of two path graphs which are circulant up to a missing edge, incorporates that the inherent vanishing property of the graph Laplacian of a circulant graph is to some extent preserved via the product operation. We investigated this phenomenon further for general circulants in \cite{splinesw}.\\
\begin{figure}[htb]\centering
\begin{minipage}{6in}
 \centering
  \raisebox{-0.5\height}{\includegraphics[height=1.2in]{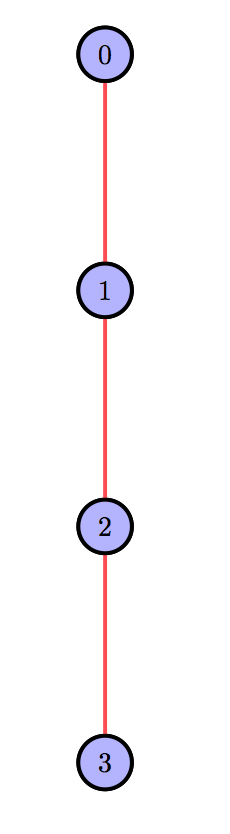}}
  \hspace*{.02in}
  ${\bf \times}$
  \hspace*{.06in}
 \raisebox{-0.4\height}{\includegraphics[height=0.2in]{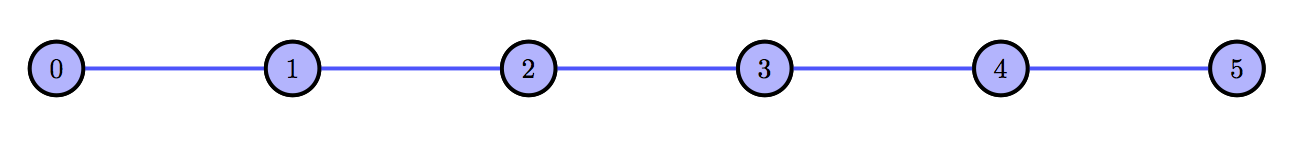}}
 \hspace*{.03in}
 {\bf =}
 \hspace*{.03in}
  \raisebox{-0.5\height}{\includegraphics[height=1.4in]{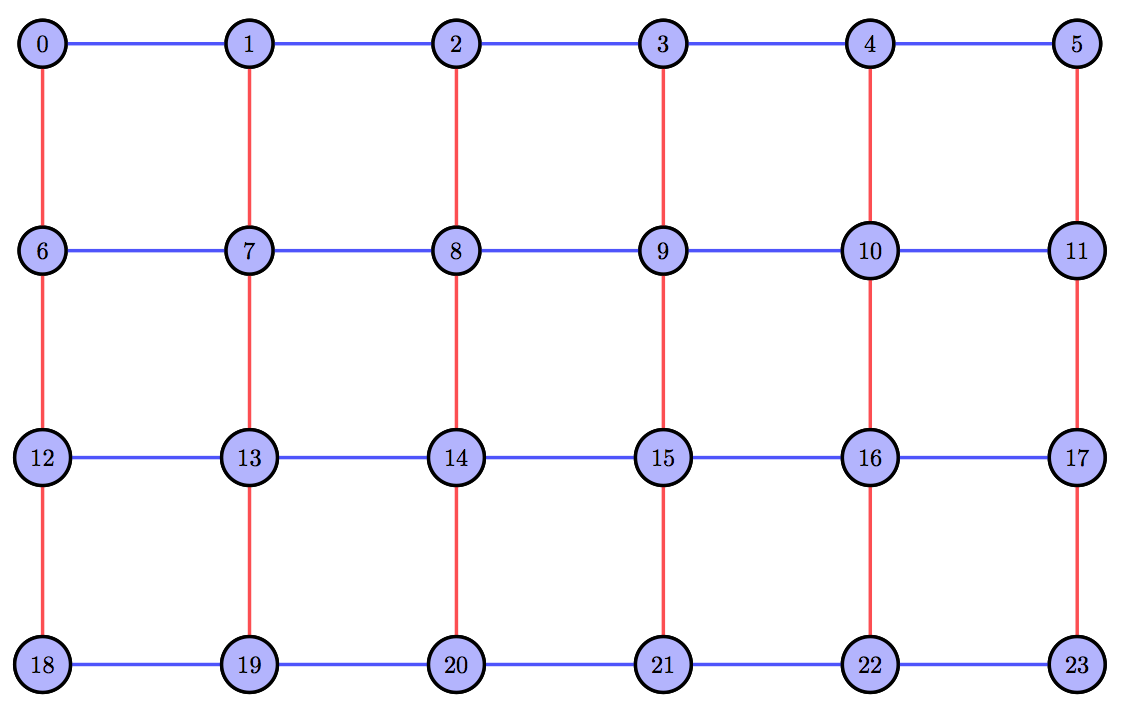}}
\end{minipage}
\caption{Graph Cartesian Product of two unweighted path graphs.}
\label{fig:gpp}
\end{figure}\\
Following the GFRI Thm. for paths, one can equivalently perform multidimensional sampling and reconstruction of signals defined on the graph product of path graphs, where the sparse tensor factors ${\bf x}_i$ on $G_i$ can be perfectly recovered based on at least $4K_i$ consecutive samples of their dimensionality-reduced GFT (DCT) representation. 
 
\subsubsection{The Kronecker product approximation}
In order to extend the GFRI framework beyond circulant and path graphs to arbitrary graphs, which can be decomposed as a graph product of the former, we resort to posing the decomposition in more general terms as an optimization problem subject to linear constraints and draw from a result in matrix theory \cite{loankron}.\\
The adjacency matrix ${\bf A}$ of an arbitrary graph $G$ may be approximated in Frobenius-norm as the Kronecker product ${\bf A}_1\otimes {\bf A}_2$ of (adjacency) matrices ${\bf A}_i$ of dimension $N_i$, which are restricted to be circulant by imposing linear constraints:
\[\min_{{\bf C}_1^T vec({\bf A}_1)=0, {\bf C}_2^T vec({\bf A}_2)=0} ||{\bf A}-{\bf A}_1\otimes {\bf A}_2||_F.\]
Here, ${\bf C}_i$ are structured, rectangular matrices with entries $\{0,1,-1\}$, which can enforce symmetry and bandedness, in addition to circularity, via the column-stacking operation by $vec$ \cite{pits}; the specific instance of a $1$-banded Toeplitz structure can similarly be imposed on either ${\bf A}_i$ to obtain a decomposition into path graphs.

\subsection{Multidimensional Separable Graph Wavelet Analysis}
At last, we briefly revisit the \textit{separable} graph wavelet transform \cite{splinesw}, which is defined on the individual circulant factors $G_i$ of product graphs, as a means to further extend the foregoing discussion on sampling. \\
\\
Let ${\bf W}_i$ denote the graph (e-)spline wavelet transform constructed in the vertex domain of circulant graph factor $G_i$, as defined in Sect. $4.1$, and ${\bf w}_i={\bf P}_{N_i} {\bf W}_i {\bf x}_i$ the graph wavelet domain representation of ${\bf x}_i$ on $G_i$, subject to the (node relabelling) permutation ${\bf P}_{N_i}$. 
For a multilevel analysis, the transform \[{\bf W}_i^{(j)}=\begin{bmatrix} {\bf W}_i^{j}&\\& {\bf I}_{N_i-\frac{N_i}{2^{j}}}\end{bmatrix}\dots {\bf W}_i^0\] induces the representation ${\bf w}_i={\bf P}^{(j)}_{N_i} {\bf W}^{(j)}_i {\bf x}_i$, with iterated permutation matrix \[{\bf P}_{N_i^{(j)}}={\bf P}_i^0\dots\begin{bmatrix} {\bf P}_i^{j}&\\& {\bf I}_{N_i-\frac{N_i}{2^{j}}}\end{bmatrix}\] at $j\leq J-1$ levels. 
Ensuing representation ${\bf w}={\bf w}_1\otimes {\bf w}_2$, which is redefined as a graph signal on $G_{\diamond}$, is the result of a separable, two-dimensional graph spline wavelet transform, as introduced in \cite{splinesw}.\\
Hence, we can similarly apply the framework of sampling and perfect reconstruction to GWT representations ${\bf w}_i$ on the vertices of $G_i$, when the given signal ${\bf x}$ is composed of smooth graph signal tensor factors ${\bf x}_i$ such that 2-D multiresolution graph wavelet representation ${\bf w}={\bf w}_1\otimes {\bf w}_2={\bf P}^{(j)}_{N_1 N_2}({\bf W}^{(j)}_{1}\otimes {\bf W}^{(j)}_{2}){\bf x}$ is $K$-sparse with $||{\bf w}_i||_0=K_i$ and $K=K_1 K_2$, for suitable graph wavelet transforms ${\bf W}^{(j)}_{i}$ and permutation matrices ${\bf P}_{N_1 N_2}^{(j)}$ at level $j\leq J-1$. Eventually the original signal ${\bf x}$ can be recovered from ${\bf w}$, subject to invertibility of the 2-D graph wavelet transform.\\

\section{Conclusion}
In this work, we have introduced a novel framework for the sampling and perfect reconstruction of sparse and wavelet-sparse graph signals with an associated graph coarsening scheme for circulant graphs, which is based on graph spline wavelet theory and can be generalised to arbitrary graphs, i.a. via graph product decomposition. Here, we have leveraged previously developed families of graph spline and graph e-spline wavelets which further facilitate the extension of the GFRI framework to the wider class of (piecewise) smooth graph signals, while establishing theoretical links to traditional sampling with a finite rate of innovation in the Euclidean domain.\\
It would be of interest to explore further sparsifying transforms on graphs within the developed sampling framework by tackling the more generalized problem of identifying a suitable transform, given an arbitrary graph and graph signal, which can induce a sparse representation.

\appendix
 \section{}
\subsection{}
\begin{proof}[Proof of Corollary 3.1]
Let ${\bf W}=\frac{1}{2}({\bf I}_N+{\bf K}){\bf H}_{LP}+\frac{1}{2}({\bf I}_N-{\bf K}){\bf H}_{HP}$ represent the generalized \textit{HGESWT}-matrix from Thm. $3.2$ with ${\bf H}_{LP/HP}=\prod_{n=1}^T\frac{1}{2^k} \left(\beta_n{\bf I}_N\pm\frac{{\bf A}}{d}\right)^k$ and diagonal downsampling matrix ${\bf K}$, with $K(i,i)=1$ at even-numbered (node) positions, and $K(i,i)=-1$ otherwise. Then we obtain \[{\bf W}^T{\bf W}=\frac{1}{2}({\bf H}_{LP}^2+{\bf H}_{HP}^2+{\bf H}_{LP}{\bf K}{\bf H}_{LP}-{\bf H}_{HP}{\bf K}{\bf H}_{HP})=\frac{1}{2}({\bf H}_{LP}^2+{\bf H}_{HP}^2)\]
where the RHS is the result of the equality ${\bf K}{\bf H}_{LP}={\bf H}_{HP}{\bf K}$ (and the equivalent in representer polynomial form $H_{HP}(z)=H_{LP}(-z)$). Thus ${\bf W}^T{\bf W}$ is circulant, i.e. it has the same basis as the circulant adjacency matrix $\frac{{\bf A}}{d}={\bf V}{\bf \Gamma}{\bf V}^H$ and its eigenvalues can be expressed as ${\bf \Lambda}=\frac{1}{2}(\prod_{n=1}^T\frac{1}{2^{2k}} \left(\beta_n{\bf I}_N+{\bf \Gamma}\right)^{2k}+\prod_{n=1}^T\frac{1}{2^{2k}} \left(\beta_n{\bf I}_N-{\bf \Gamma}\right)^{2k})$. Hence, the condition number of ${\bf W}$ is given by $C=\sqrt{\frac{\lambda_{max}}{\lambda_{min}}}$.

\end{proof}

\subsection{}
\begin{proof}[Proof of Corollary 4.1]

$(i)$ The number of non-zero high-pass coefficients after applying one level of the HGSWT is $B$; due to the additional `border effect' of the low-pass filter at subsequent levels, we obtain the following series after $j$ levels
\[S=B+\left(B+\frac{B}{2}\right)+\left(B+\frac{\frac{B}{2}+B}{2}\right)+\dots=\sum_{n=0}^{j-1}(j-n)\frac{B}{2^n}.\]
Using the finite summation results \[\sum_{n=0}^{j-1}\frac{j}{2^n}=j(2-2^{-j+1}),\quad\text{and}\quad\sum_{n=0}^{j-1}\frac{n}{2^n}=2^{(1-j)}(-j-1+2^j)\]
and considering the  $\frac{N}{2^j}$ low-pass coefficients, we obtain $K=\frac{N}{2^j}+B(2(j-1)+2^{1-j})$ as the total number of non-zeros. Here, for large $B$, the number of high-pass coefficients at each level $l\leq j-1$ is bounded $\sum_{n=0}^{l}\frac{B}{2^n}\leq \frac{N}{2^{l+1}}$. If $B= 2^{j-1} r$, the formula for $S$ is exact, otherwise, since $S$ has to be an integer, we need to adjust the formula by adding/subtracting a term $s_l$ at each level $l$, depending on whether downsampling requires rounding up or down. In particular, at each level, the high-pass filter is applied on the odd-numbered nodes $1,3,...$ of the (previously) low-pass filtered and sampled graph signal $\tilde{{\bf y}}$; we thus note that if the length of the non-zero `border' support (before downsampling) of the high-pass filtered $\tilde{{\bf y}}$ at the beginning of the resulting labelled sequence is an even number, while that at the end of the sequence is odd, we need to round up, and vice versa.\\
\\
\noindent $(ii)$ Following the reasoning of the previous proof, we need to consider the border effect caused by filtering with the low-pass filter of bandwidth $T$; we therefore end up with the following series summation for the total number of non-zeros
\[K=\frac{N}{2^j}+B+\left(B+\frac{T}{2}\right)+\left(B+\frac{\frac{T}{2}+T}{2}\right)+\dots=\frac{N}{2^j}+j B+\sum_{n=1}^{j-1}\frac{T (j-n)}{2^n},\]
giving the formula $K=\frac{N}{2^j}+j B+ T(j-2+2^{1-j})$ subject to a correction term $\pm s_l$ per level.\\
\\
\noindent $(iii)$ By Thm $3.1$, we need to retain at least one low-pass component for invertibility of the filterbank, therefore we choose to assign the low-pass component to only one node, while the remaining nodes retain the high-pass components. While this downsampling approach is not conducted with respect to the generating set of the circulant graph, and therefore less rigorous from a graph-theoretical perspective, it achieves a maximally sparse representation in the graph wavelet domain.
The number of non-zeros are $2B$, where $2B-1$ is the number of non-zero high-pass coefficients.\end{proof}

\subsection{}
\begin{proof}[Proof of Thm. 4.1](Prony's method \cite{drag}):\\
Given the representation ${\bf y}={\bf U}_M^H{\bf x}$, where $||{\bf x}||_0=K$ and ${\bf U}_M^H$ are the first $M$ rows of the DFT-matrix, we can represent the n-th entry of ${\bf y}$ as $y_n=\frac{1}{\sqrt{N}}\sum_{k=0}^{K-1}x_{c_k}e^{-i2\pi c_k n/N}$ with weights $x_{c_k}$  of ${\bf x}$ at positions ${c_k}$, and apply Prony's method to recover ${\bf x}$, provided $M\geq 2K$. Here, we redefine $y_n=\sum_{k=0}^{K-1}\alpha_k u_k^{n}$ with locations $u_k=e^{-i2\pi c_k n/N}$ and amplitudes $\alpha_k=x_{c_k}/\sqrt{N}$, which are successively recovered.
In the following, we summarise the reconstruction algorithm: given the samples $y_n$, we construct a Toeplitz matrix ${\bf T}_{K,l}$, and determine the vector ${\bf h}$, which lies in its nullspace, also known as `the annihilating filter': 
\[{\bf T}_{K,l}{\bf h}=\begin{pmatrix}
y_{l+K}&y_{l+K-1}&\dots&y_{l}\\
y_{l+K+1}&y_{l+K}&\dots&y_{l+1}\\
\vdots&\ddots&\ddots&\vdots\\
y_{l+2K-2}&\ddots&\ddots&\vdots\\
y_{l+2K-1}&y_{l+2K-2}&\dots&y_{l+K-1}
\end{pmatrix}\begin{pmatrix}
1\\
h_1\\
h_2\\
\vdots\\
h_K
\end{pmatrix}={\bf 0}_K
\]
which can be accomplished via the SVD-decomposition of ${\bf T}_{K,l}$. It can be shown that ${\bf T}_{K,l}$ is of rank $K$ for distinct $u_k$ (Prop. 1, \cite{drag}). In particular, this corresponds to the matrix-form expression of $\sum_{0\leq k\leq K-1}\alpha_k u_k^n P(u_k)=0$ for $l\leq n <l+K$, with polynomial \[P(x)=x^K+\sum_{k=1}^K h_k x^{K-k}=\prod_{k=1}^{K}(x-u_{k-1})\] whose roots $\{u_k\}_{k=0}^{K-1}$ can be subsequently determined from ${\bf h}$. At last, we can recover the corresponding amplitudes $\{\alpha_k\}_{k=0}^{K-1}$ by solving a system of $K$ linear equations given by $y_n$.
\end{proof}

\subsection{}
\begin{proof}[Proof of Lemma 4.2] (appears in part in \cite{icassp}).\\
The eigenvalues of ${\bf A}$ with first row $\lbrack 0\quad a_1...a_1\rbrack$ are
$\lambda_j=\sum_{k=1}^{B} 2 a_k cos\left({\frac{2\pi j k}{N}}\right),\enskip j=0,...,N$-$1$.
Thus the eigenvalues of $\tilde{{\bf A}}$ with the same entries $a_i$ and bandwidth $B<N/4$, are
$\tilde{\lambda}_j=\sum_{k=1}^{B} 2 a_k cos\left({\frac{2\pi (2j) k}{N}}\right)=\lambda_{2j},\enskip j=0,...,N/2$-$1$. 
We can similarly show the preservation of the downsampled DFT-eigenbasis. Let
  \[{\bf x}=\begin{bmatrix} e^{i\alpha 0}&
 e^{i\alpha 1}&
 e^{i\alpha 2}&
 \dots &
 e^{i\alpha (N-1)}\end{bmatrix}^T\]
 with $\alpha=-\frac{2\pi k}{N}$ denote the $(k+1)$-th row of the non-normalized DFT-matrix. If we discard all entries at odd-numbered positions, we obtain the $(k+1)$-th row of the DFT of dimension $N/2$, since
 \[\begin{bmatrix} e^{i\alpha 0}&
 e^{i\alpha 2}&
 e^{i\alpha 4}&
 \dots &
 e^{i\alpha (N-2)}\end{bmatrix}^T=\begin{bmatrix} e^{i(2\alpha) 0}&
 e^{i(2\alpha) 1}&
 e^{i(2\alpha) 2}&
 \dots &
 e^{i(2\alpha) (N/2-1)}\end{bmatrix}^T\]
with $2\alpha=-\frac{2\pi k}{N/2}$. Thus, if we apply the above sampling pattern on the first $N/2$ rows of the DFT of dimension $N$, we obtain the DFT of dimension $N/2$. In particular, at $k=N/2$, we have $\alpha=\pi$ and thus corresponding, downsampled row ${\bf x}={\bf 1}_{N/2}$, and proceeding similarly, we observe that the sampled lower half of the $N\times N$ DFT equivalently gives the DFT of dimension $N/2$.

\end{proof}






\bibliographystyle{elsarticle-num} 
\bibliography{sample2}


\end{document}